\documentclass[lettersize,journal]{IEEEtran}
\usepackage{times}
\usepackage{setspace}
\usepackage[utf8]{inputenc} 
\usepackage[T1]{fontenc}
\usepackage{url}
\usepackage{ifthen}
\usepackage{cite}
\usepackage[cmex10]{amsmath}
\usepackage{booktabs}       
\usepackage[flushleft]{threeparttable}
\usepackage{amsfonts}       
\usepackage{nicefrac}       
\usepackage{microtype}      
\usepackage{amsthm,amsmath,amsfonts, amssymb}
\usepackage{multirow}
\usepackage{graphicx}
\usepackage{xcolor}
\usepackage{mathrsfs,mathtools}
\usepackage{bm, bbm}
\usepackage[algo2e]{algorithm2e}
\usepackage{wrapfig}
\usepackage{hyperref}
\usepackage{cleveref}
\usepackage{todonotes}
\usepackage{adjustbox}

\definecolor{commentgreen}{RGB}{106,153,85}

\RequirePackage{algorithm}
\RequirePackage{algorithmic}

\theoremstyle{plain}

\newtheorem{ass}{Assumption}
\newtheorem{thm}{Theorem}

\newtheorem{cor}{Corollary}
\newtheorem{lemma}{Lemma}
\newtheorem{remark}{Remark}
\newtheorem{ex}{Example}

\newcommand*{\eg}{\emph{e.g.}{}}
\newcommand*{\ie}{\emph{i.e.}{}}

\DeclareMathOperator*{\esssup}{esssup}
\renewcommand{\d}{\mathrm{d}}   
\newcommand{\w}{w}     
\newcommand{\alg}{\texttt{ST-Score}} 

\begin{document}

\title{Score-Based Change-Point Detection and Region Localization for Spatio-Temporal Point Processes} 




\author{Wenbin Zhou, Liyan Xie, and Shixiang Zhu%
\thanks{W. Zhou and S. Zhu are with the Heinz College of Information Systems and Public Policy at Carnegie Mellon University, Pittsburgh, PA 15213 USA (e-mails: \url{wenbinz2@andrew.cmu.edu}; \url{shixiangzhu@cmu.edu}).}
\thanks{L. Xie is with the Department of Industrial Systems and Engineering at the University of Minnesota, Minneapolis, MN 55455 USA (e-mail: \url{liyanxie@umn.edu}).}
}

\maketitle

\begin{abstract}
We study sequential change-point detection for spatio-temporal point processes, where actionable detection requires not only identifying when a distributional change occurs but also localizing where it manifests in space. While classical quickest change detection methods provide strong guarantees on detection delay and false-alarm rates, existing approaches for point-process data predominantly focus on temporal changes and do not explicitly infer affected spatial regions. We propose a likelihood-free, score-based detection framework that jointly estimates the change time and the change region in continuous space-time without assuming parametric knowledge of the pre- or post-change dynamics. The method leverages a localized and conditionally weighted Hyv\"arinen score to quantify event-level deviations from nominal behavior and aggregates these scores using a spatio-temporal CUSUM-type statistic over a prescribed class of spatial regions. Operating sequentially, the procedure outputs both a stopping time and an estimated change region, enabling real-time detection with spatial interpretability. We establish theoretical guarantees on false-alarm control, detection delay, and spatial localization accuracy, and demonstrate the effectiveness of the proposed approach through simulations and real-world spatio-temporal event data.
\end{abstract}


\section{Introduction}

Change-point detection is a classic and central problem in statistical signal processing, concerned with the real-time detection of abrupt changes in the probabilistic structure of sequential observations \cite{poor-hadj-QCD-book-2008}. Given a data stream, the objective is to identify the time at which the underlying distribution deviates from its nominal regime with minimal detection delay, subject to false-alarm constraints. Such problems arise in a wide range of signal processing applications, including seismic monitoring \cite{xie2019asynchronous}, industrial process control \cite{shi2009quality}, network surveillance, and supply-chain monitoring \cite{yamin2022online}. In these settings, rapid and reliable detection is critical for enabling timely intervention and mitigating downstream impacts.

More recently, many modern signal processing systems generate \emph{event-type spatio-temporal data}, consisting of discrete events indexed jointly by time and spatial location. Such data are naturally modeled using spatio-temporal point processes \cite{reinhart2018review, daley2006introduction} and arise in applications such as earthquake catalogs \cite{omi2014estimating}, wildfire ignitions \cite{xu2023spatio}, crime events \cite{park2021investigating, zhu2022spatiotemporal}, and traffic incidents \cite{zhu2021spatio}. Existing change-point detection methods for point processes have primarily focused on temporal changes, aiming to detect shifts in event intensity or dynamics over time while treating space as either fixed or ancillary \cite{xie2019asynchronous, liu2015adaptive}. While this perspective is effective for detecting global distributional shifts, it abstracts away spatial structure that is often central to interpretation and response.

In many real-world applications, however, the significance of a change is inseparable from \emph{where} it occurs. In seismic monitoring, for instance, elevated seismicity is routinely interpreted through its spatial concentration---such as activity localized along a fault segment or within a confined region---since such patterns distinguish meaningful seismic sequences (\eg, swarm-like activity or foreshocks) from diffuse background fluctuations \cite{horalek2015earthquake}. Detecting a change in aggregate event rates without spatial localization offers limited geophysical insight and cannot support downstream analyses of fault behavior or hazard evolution. A similar consideration arises in wildfire monitoring, where changes in ignition patterns are operationally meaningful only when they can be spatially resolved, allowing agencies to identify high-risk areas, prioritize surveillance, and allocate mitigation resources \cite{romero2008gis, khan2024investigation}. These examples underscore that actionable change detection in spatio-temporal event systems requires jointly inferring when a change occurs and where it manifests, rather than treating change-point detection as a purely temporal problem.

Despite its practical importance, methods that jointly perform sequential change-point detection and localize the affected spatial region remain comparatively less developed than their purely temporal counterparts, especially for spatio-temporal point-process data. This limitation is primarily driven by two modeling and computational challenges:
($i$) In spatio-temporal point processes, a ``change'' can take multiple, non-equivalent forms, \eg, a shift in the background intensity surface \cite{wang2023sequential}, the emergence of a spatially localized increase in event activity \cite{reinhart2018review}, or a change in interaction or triggering structure \cite{dong2023non}, as in self- or mutually-exciting processes. As a result, the post-change behavior cannot be characterized by a finite set of alternatives; instead, it belongs to an infinite-dimensional family of possible post-change scenarios. In contrast, much of the existing literature on change detection and localization focuses on identifying the post-change scenario from a finite set of candidates \cite{nikiforov2003lower}. Such approaches are therefore not directly applicable in our setting.
($ii$) From an inferential standpoint, the object of interest is a latent spatio-temporal intensity (or conditional intensity) evolving in continuous space-time, while observations arrive only as sparse, irregular event times and locations. This makes fully likelihood-based sequential procedures delicate: they are often tractable only under strong parametric assumptions (\eg, Poisson/Hawkes with specified structure), and become substantially harder when the pre-/post-change dynamics are unknown, high-dimensional, or spatially heterogeneous---precisely the regimes where region identification is most valuable \cite{wang2023sequential}. 

To address these challenges, we propose a likelihood-free, score-based framework for spatio-temporal change-point detection that simultaneously identifies the change time and the affected spatial region. Our approach avoids explicit likelihood estimation and does not assume prior knowledge of the pre- or post-change distributions. Instead, it leverages the score function---the gradient of the log-density---which fully characterizes the underlying distribution and can often be estimated more reliably than the likelihood itself in high-dimensional settings \cite{zhou2025sequential}. 

At the core of our method is a localized and conditional weighted Hyv\"arinen score that quantifies the anomalousness of each observed event relative to pre-change dynamics. These event-level scores are then aggregated using a spatio-temporal CUSUM-type statistic over a prespecified class of spatial regions. Operating sequentially, the algorithm searches for regions whose cumulative scores exceed a detection threshold; when such a region is identified, the procedure outputs both a stopping time and an estimated change region, enabling real-time detection with spatial interpretability.
We establish theoretical guarantees for the proposed algorithm along three dimensions.
First, we characterize its false alarm behavior by deriving lower bounds on the average run length.
Second, we analyze the detection efficiency and obtain bounds on the detection delay, characterizing how quickly the algorithm responds after a change occurs.
Third, we provide theoretical guarantees on change-region localization accuracy, quantified using the Jaccard index between the estimated and true change regions.

Our contributions are threefold.
($i$) We introduce a continuous spatio-temporal change-region detection framework for point processes, where both the change-point and the affected spatial region are unknown, bridging the gap between classical time-only CUSUM methods and ad hoc spatial scanning approaches.
($ii$) We develop a scalable, likelihood-free detection method based on score matching that jointly detects change-points and localizes affected regions via a novel spatio-temporal CUSUM-type statistic and an alternating optimization scheme.
($iii$) We establish asymptotic guarantees linking detection delay and localization accuracy, and demonstrate the effectiveness of the proposed approach through comprehensive simulations and two real-world applications.

The remainder of the paper is organized as follows.
\Cref{sec:related} reviews related work; \Cref{sec:prelim} introduces the problem formulation and preliminaries; \Cref{sec:method} details the proposed method; \Cref{sec:theory} provides theoretical analysis of the algorithm; and \Cref{sec:exp} reports numerical results.

\section{Related Works}
\label{sec:related}

We review three lines of prior work most relevant to our study:
($i$) classical sequential change detection and multi-stream localization,
($ii$) change detection for spatio-temporal point processes,
and ($iii$) score-based methods for density and change-point modeling.

\paragraph{Classic Change Detection and Localization} The classic version of sequential change-point detection for independent and identically distributed (i.i.d.) observations with known distribution is a well-developed area. For a review, see books such as \cite{poor-hadj-QCD-book-2008,Siegmund1985,tartakovsky2014sequential} or papers such as \cite{Lai:2001,tutorial_jsait}. 
In the past decade, sequential change detection of multiple data streams has received much attention, such as detection across high-dimensional independent streams \cite{xie:2013,mei2010efficient,Wang:2015,chan2017optimal} and correlated data streams \cite{xie2020subspace,yan2018real,keshavarz2020sequential,chen2022high}. One line of work that is partially related to our study is change-point detection and localization in multi-stream settings, where changes occur in an unknown subset of data streams and the goal is to both detect the change point and identify the affected streams. In a series of works, Nikiforov introduced a minimax optimal detection-isolation algorithm for stochastic dynamical systems \cite{nikiforov1995generalized}, developed a recursive variant of the algorithm that achieves better computational efficiency \cite{nikiforov2000simple}, and provided an asymptotic lower bound for the mean detection delay with constraints on the probability of false localization and the average time before a false alarm \cite{nikiforov2003lower}. See  \cite{tartakovsky2008multidecision,tartakovsky2014sequential,tartakovsky2019sequential} for more detailed overviews. It is worth mentioning that those existing work on detection and localization typically assumes a finite set of post-change scenarios, such as changes affecting one of finitely many data streams. In contrast, we consider a continuous spatial change region and aim to identify the region jointly with change detection.

\paragraph{Spatio-Temporal Point Processes and Their Change Detection}
Point process models are widely used for modeling discrete event data due to their ability to directly characterize inter-event waiting times \cite{daley2007introduction}. In this framework, Poisson processes assume independent exponential inter-arrival times, while Hawkes processes introduce temporal dependence by allowing the conditional intensity to depend on the history of past events \cite{coevolve2015}. This self-exciting structure makes Hawkes processes particularly suitable for capturing temporal dependence and causal effects, with applications ranging from financial markets \cite{Toke2010} and earthquake modeling \cite{ogata1988statistical} to social and information networks \cite{fox2014modeling, HawkesTopic15}. Multi-dimensional Hawkes processes further extend this modeling capacity to networked settings, enabling the representation of strongly correlated event streams and signal propagation over networks \cite{reinhart2018review}.
Change detection in point process models has attracted sustained attention in both single-stream and multi-stream settings. Existing work includes classical treatments for Poisson processes \cite{shen2012change, zhang2016scanning, herberts2004optimal}, as well as more recent developments for one-dimensional \cite{Ludkovski12,pinto2015trend} and multi-dimensional or network-based point processes \cite{li2017detecting,wang2023sequential}. Notably, \cite{li2017detecting} proposed GLR-type procedures for detecting changes in networked event streams, while \cite{wang2023sequential} developed a penalized dynamic programming method for identifying parameter changes in high-dimensional self-exciting Poisson processes in an offline setting. This literature is also closely related to multisource quickest detection problems, which typically assume independence across data streams; representative examples include optimal Bayesian procedures for detecting the minimum of multiple change-points in compound Poisson models \cite{bayraktar2007quickest} and optimal rules for coupled systems governed by It\^o processes \cite{zhang2014quickest}.
Our work is related to these studies in that we assume a spatio-temporal Hawkes process event dynamic. However, our detection problem differs: we aim to detect changes jointly in time and continuous space, whereas most existing approaches focus on the discretized multi-stream settings.

\paragraph{Score-Based Methods}

The fundamental work on score matching was first introduced in \cite{hyvarinen2005estimation} as an effective method to estimate unnormalized score models.
Its key idea is to minimize the mean squared error between the score model and the true (Stein) score function \cite{hyvarinen2005estimation}, where the latter can be factored out using the integrate-by-parts technique.
However, this trick relies on the assumption that the density is supported on an unconstrained domain, which restricts its application in sequential modeling domains \cite{meng2020autoregressive}.
To address this limitation, weighted score matching was introduced \cite{hyvarinen2007some, yu2019generalized, liu2022estimating}, whose core idea is to construct reweighted densities that would meet the assumptions.
A complementary line of work studying denoising score matching \cite{vincent2011connection} reframes the estimation problem as a denoising task, where the score model attempts to predict the injected noise from the perturbed data and establishes the equivalence between it and the original score matching when the injected noise has zero variance.
Later works show that denoising score matching also provides improved numerical stability and estimation efficiency, particularly in high-dimensional or manifold-structured data, leading to a series of developments in score-based diffusion models \cite{song2019generative, song2020score, ho2020denoising}.
Recently, score-based methods have also been applied to point process modeling \cite{li2023smurf, cao2024score, dong2025conditional}, where \cite{dong2025conditional} proposed a general denoising score matching-based model for modeling the conditional PDF of point processes.
Another recent line of work studies score-based methods in change-point detection problems \cite{wu2024quickest} by adapting the Hyv\"arinen score as its detection statistics (SCUSUM). Several extensions have been developed, including the
Bayesian setting \cite{banerjee2024bayesian},
robust setting \cite{moushegian2025robust},
diffusion-based setting \cite{moushegian2025diffusion},
conditional setting \cite{chen2025conditional},
denoising score matching estimation \cite{zhou2025sequential},
and multistream setting \cite{chen2025score}.
Our work is most closely related to these studies of score-based change detection, but differs in its focus on point process event-type sequences rather than time series.

\section{Problem Setup and Preliminaries}
\label{sec:prelim}

In this section, we review spatio-temporal point processes in \Cref{sec:mtpp}, formulate the change-point detection and region localization problem in \Cref{sec:setting}, and revisit the classic CUSUM test, emphasizing its limitations in spatio-temporal settings in \Cref{sec:cusum}.

\subsection{Spatio-Temporal Point Processes}
\label{sec:mtpp}

A \emph{spatio-temporal point process} (STPP) is a stochastic model for discrete events occurring in time and space \cite{reinhart2018review}. A realization over a finite horizon $[0,T)$ is represented by the event history
\[
\mathcal{H}_T = \{x_1,\ldots,x_{N_T}\} \subseteq [0,T)\times\mathcal S \eqqcolon \mathcal X,
\]
where each event $x_j=(t_j,s_j)$ consists of a time $t_j\in[0,T)$ and a spatial location $s_j\in\mathcal S$. We use $t(x)$ and $s(x)$ to denote the temporal and spatial components of any $x\in\mathcal X$, respectively, and $N_T$ denotes the total number of events observed before time $T$.

The process can also be characterized through its counting measure $\mathbb N$, defined for any measurable subset $\mathcal X'\subseteq\mathcal X$ as
\[
\mathbb N(\mathcal X') = \lvert \mathcal H_T \cap \mathcal X' \rvert,
\]
which counts the number of events falling in $\mathcal X'$. The local dynamics of an STPP are governed by its \emph{conditional intensity function}, which specifies the instantaneous event rate at $x=(t,s)$ given the history up to time $t$:
\begin{equation}
\label{eq:intensity}
\lambda(x \mid \mathcal H_t)
= \lim_{|\mathrm d x|\to 0}
\frac{\mathbb E\left[\mathbb N(\mathrm d x)\mid \mathcal H_t\right]}{\mathrm d x}.
\end{equation} 
The conditional intensity fully characterizes the distribution of future events. In particular, letting $t_n$ denote the time of the most recent event prior to $t(x)$, the conditional density of the next event occurring at $x$ is given by
\begin{align*}
p(x \mid \mathcal H_{t(x)})
=&~
\lambda(x\mid\mathcal H_{t(x)}) \cdot \\
&~\exp\left(-\int_{[t_n,t(x))\times\mathcal S}\lambda(u\mid\mathcal H_{t(u)})\mathrm du\right).
\end{align*}

Given an observed event stream $\mathcal H_T$, the log-likelihood of a model with intensity $\lambda$ admits the following form:
\begin{equation}
\label{eq:ll}
\ell(\mathcal H_T)
= \int_{\mathcal X} \log \lambda(x | \mathcal H_{t(x)})\mathrm d\mathbb N(x) - \int_{\mathcal X} \lambda(x | \mathcal H_{t(x)})\mathrm d x,
\end{equation}
where integration with respect to the counting measure corresponds to summation over observed events, \ie,
\[
\int_{\mathcal X'} \phi(x)\mathrm d\mathbb N(x)
= \sum_{x_j\in\mathcal H_T\cap\mathcal X'} \phi(x_j)
\quad
\text{for any } \phi:\mathcal X\to\mathbb R.
\]

\subsection{Problem Setup for Spatio-Temporal Detection}
\label{sec:setting}

\begin{figure}[t!]
    \centering
    \includegraphics[width=1.0\linewidth]{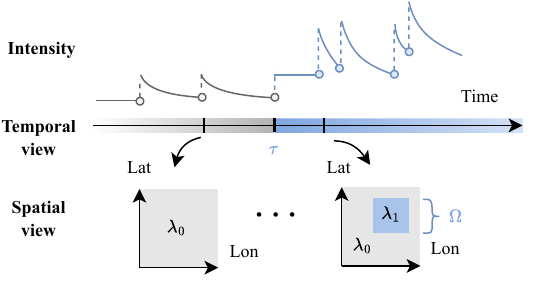}
    \vspace{-0.2in}
    \caption{Illustration of the problem setting. After some time $\tau$, the intensity function of the event generation process shifts from $\lambda_0$ to $\lambda_1$ within a confined spatial region $\Omega$. 
    }
    \label{fig:setup}
    \vspace{-0.1in}
\end{figure}

We consider a stream of events generated by an STPP with conditional intensity defined as in \eqref{eq:intensity}. Our goal is to determine {\it whether}, and {\it where}, the underlying intensity undergoes a structural change. Specifically, we seek to detect a shift from a pre-change intensity $\lambda_0(\cdot)$ to a post-change intensity $\lambda_1(\cdot)$ occurring at an \emph{unknown} change-point in time $\tau$ and confined to an \emph{unknown} spatial region $\Omega$.

Recall that $\mathcal X \coloneqq [0,T)\times\mathcal S \subseteq \mathbb R_+\times\mathbb R^2$ denotes the full spatio-temporal domain. The change, if present, is assumed to affect only the post-change region $\mathcal G_{\rm post} \coloneqq [\tau,T)\times\Omega$,
while the process remains unchanged on its complement, referred to as the pre-change region $\mathcal G_{\rm pre}\coloneqq \mathcal X\setminus\mathcal G_{\rm post}$. Under this formulation, the detection task can be cast as a hypothesis testing problem:
\begin{equation}
\label{eq:test}
\begin{aligned}
H_0: \quad
& \lambda(x\mid\mathcal H_{t(x)}) = \lambda_0(x\mid\mathcal H_{t(x)}),
\qquad \forall x\in\mathcal X,\\
H_1: \quad
& \lambda(x\mid\mathcal H_{t(x)}) = \lambda_0(x\mid\mathcal H_{t(x)}),
\qquad \forall x\in\mathcal G_{\rm pre},\\
& \lambda(x\mid\mathcal H_{t(x)}) = \lambda_1(x\mid\mathcal H_{t(x)}),
\qquad \forall x\in\mathcal G_{\rm post}.
\end{aligned}
\end{equation}
An illustration of this spatio-temporal change-point setting is given in \Cref{fig:setup}. For notational simplicity, we omit the explicit dependence on the history $\mathcal H_{t(x)}$ in our notations for the remainder of the paper.
Note that our problem setup extends naturally to marked spatio-temporal point processes \cite{hawkes1974cluster, ogata1988statistical} by augmenting the spatio-temporal domain $\mathcal{X}$ with a mark space.

The hypothesis formulation in \eqref{eq:test} is broad and encompasses many commonly used point process models. We highlight two representative examples below.

\begin{ex}[Homogeneous Poisson Process]
\label{ex:hpp}
Suppose both the pre- and post-change dynamics follow homogeneous Poisson processes, with constant intensities $\mu_0\neq\mu_1$. In this case, $\lambda_i(x)=\mu_i$ for $i=0,1$.
\end{ex}

\begin{ex}[Self-exciting Hawkes Process]
\label{ex:hawkes}
A more flexible class is given by Hawkes processes \cite{hawkes1974cluster}, whose conditional intensity before and after the change takes the form
\begin{equation}
\label{eq:hawkes}
\lambda_i(x)
=
\mu_i
+
\alpha
\sum_{x'\in\mathcal H_{t(x)}} \kappa(x,x'),
\qquad i=0,1,
\end{equation}
where $\mu_i\in\mathbb R_{\ge 0}$ is the base intensity and $\kappa:\mathcal X\times\mathcal X\to\mathbb R_{\ge 0}$ is a self-excitation kernel and $\alpha$ controls the strength of the triggering effect. A common choice is the exponential kernel,
\[
\kappa(x,x')=\beta \exp\big(-\beta\|x-x'\|\big),
\]
which captures localized temporal and spatial clustering effects; note that the kernel integrates to one over $x$.
\end{ex}

In \Cref{ex:hawkes}, even though the change is assumed to occur only in the base intensity (\ie, $\mu_0\neq\mu_1$), the self-exciting structure substantially complicates detection, as pre-change events may continue to affect future intensities through the accumulated history, yielding delayed, smoothed, or spatially propagated signatures that blur the true onset and spatial extent of the change. At the same time, this dependence on past events underpins the practical importance of Hawkes processes, since self-excitation is intrinsic to many real-world phenomena---such as seismic aftershocks \cite{ogata1988statistical}, crime and social interactions \cite{zhu2022spatiotemporal}, and information or epidemic diffusion \cite{dong2023non}---making them a canonical and interpretable model for spatio-temporal data with contagion or cascade effects.

\subsection{The Classic CUSUM Procedure}
\label{sec:cusum}

The cumulative sum (CUSUM) procedure, originally proposed in \cite{page-biometrica-1954}, is a classic and powerful method for sequential change-point detection. Its central principle is to accumulate evidence for a distributional change and to declare a change once this evidence becomes sufficiently strong. In the discrete-time setting, for a stream of observations $\{x_1,x_2,\ldots\}$, the CUSUM statistic at time $t$ is defined as
\[
W_t^{\rm CUSUM}
=
\max_{1\le k\le t}
\sum_{j=k}^t
\log\frac{p_1(x_j)}{p_0(x_j)},
\qquad t=1,2,\ldots,
\]
where $p_0(\cdot)$ and $p_1(\cdot)$ denote the pre- and post-change probability density functions, respectively. An alarm is triggered when the statistic exceeds a pre-specified threshold $\gamma$, leading to the stopping time
\[
\nu
=
\inf\bigl\{t\ge 1: W_t^{\rm CUSUM}\ge \gamma\bigr\}.
\]

In spatio-temporal point process settings, observations arrive in continuous time, and the likelihood in \eqref{eq:ll} can be evaluated at any $t\in\mathbb R_+$. When assuming spatial change region $\Omega$ is known and $\Omega = \mathcal{S}$, the discrete-time CUSUM statistic admits a natural continuous-time extension \cite{wang2023sequential}:
\begin{equation}
\label{eq:cusum-continuous}
W_t^{\rm CUSUM}
=
\sup_{0\le \tau\le t}
\bigl[
\ell_{1,\tau}(\mathcal H_t)
-
\ell_0(\mathcal H_t)
\bigr],
\end{equation}
where $\ell_{1,\tau}(\mathcal H_t)$ denotes the log-likelihood of the observed history up to time $t$ under the alternative hypothesis in \eqref{eq:test} with a hypothesized change-point $\tau$, and $\ell_0(\mathcal H_t)$ is the log-likelihood under the null hypothesis of no change.

Despite its appeal, the classic CUSUM framework faces two fundamental challenges in our STPP setting: ($i$) The change region $\Omega$ is unknown, a subset of $\mathcal{S}$, and must be inferred jointly with detection, rendering the post-change likelihood $\ell_{1,\tau}$ unavailable a priori and necessitating a search over candidate regions. ($ii$) The inclusion of a spatial dimension substantially increases the computational burden of evaluating region-specific log-likelihoods, motivating the development of new detection statistics that are both statistically sound and computationally tractable.

\section{Methodology}
\label{sec:method}

This section proposes a score-based CUSUM-type statistic for the spatio-temporal change-point detection problem in \eqref{eq:test}, building on the sequential score-based framework of \cite{zhou2025sequential}. 
Our approach combines a localized score approximation with an alternating minimization procedure to jointly detect the change-point and localize the affected spatial region.

\subsection{Score-Based Spatio-Temporal CUSUM Procedure}
\label{sec:stat}

For each observed event $x$, we define $\psi_i(x)$ as a regime-specific score that assesses how well $x$ aligns with the dynamics of regime $i\in\{0,1\}$.
Scores are evaluated on a transformed event $(t(x)-t_n,\, s(x))$, which encodes local inter-arrival dynamics.
The difference $\Delta(x) \coloneqq \psi_0(x)-\psi_1(x)$ serves as a local anomaly measure, assigning positive weight to events that are more compatible with the post-change regime than with the pre-change regime. 
The detailed construction and computation of the score functions $\psi_i(x)$ are described in \Cref{sec:score}.

This difference defines the detection statistic at time $t$ as:
\begin{equation}
\label{eq:stat}
W_t
=
\sup_{0 \le \tau \le t}
\sup_{\Omega \subseteq \mathcal S}
\int_{[\tau,t)\times\Omega}
\Delta(x)
\mathrm d\mathbb N(x).
\end{equation}
The statistic aggregates localized anomaly evidence over all candidate change-points $\tau$ and spatial regions $\Omega$.
The supremum localizes the spatio-temporal region in which the discrepancy between the two regimes is most pronounced.
In contrast to the classic CUSUM statistic in \eqref{eq:cusum-continuous}, which maximizes only over discrete time indices, \eqref{eq:stat} performs a joint optimization over continuous time and space, enabling detection at arbitrary spatio-temporal scales.

An alarm is raised once the statistic exceeds a prescribed threshold $\gamma$, yielding the stopping time
\begin{equation}
\label{eq:stopping-time}
\nu
=
\inf\{t: W_t \ge \gamma\}.
\end{equation}
At the stopping time, the maximizers $(\hat\tau,\hat\Omega)$ attaining $W_\nu$ in \eqref{eq:stat} provide estimates of the change-point and the affected spatial region.

\subsection{Localized Score Approximation}\label{sec:score}


We construct $\psi_i(x)$ using the weighted Hyv\"arinen score, a proper scoring rule that depends only on the score function $\nabla_x \log p_i(x)$ rather than the likelihood itself \cite{hyvarinen2005estimation}. This makes it particularly well-suited to spatio-temporal point processes, where likelihood-based approaches require evaluating intractable space-time integrals, and enables efficient, unnormalized modeling of event dynamics without computing normalizing constants \cite{li2023smurf}.

Formally, we define
\begin{equation}
\label{eq:hyvarinen}
\psi_i(x)
\coloneqq
\left\|
\sqrt{w(x)} \odot f_i(x)
\right\|_2^2
+
2\operatorname{div}
\bigl[
w(x) \odot f_i(x)
\bigr],
\end{equation}
where $f_i(\cdot)$ is a learned score model that captures $\nabla_x \log p_i(x)$ and $w(\cdot)$ is a weighting function.
We choose
\[
w(x)
\coloneqq
\min_{x' \in \mathcal X^c}
\|x' - x\|_\infty,
\]
which downweights regions near the boundary of the truncated domain $[0, T)$; more general constructions are provided in \Cref{app:drift}.
This weighting is essential because the support of inter-arrival times is one-sided, violating the regularity conditions required for unweighted score matching and integration-by-parts arguments \cite{liu2022estimating}.

A key challenge in estimating the score functions $f_i$ is that their computation depends on the unknown change region $\Omega$.
Recall that the conditional density of an event $x$ takes the form
\begin{equation}
\label{eq:pdf-s}
p(x)
=
\lambda(x)
\exp\left(
-\int_{[t_n,t(x))\times \mathcal{S}}
\lambda(u)\d u
\right),
\end{equation}
where the integral is taken over the entire spatial domain $\mathcal S$. 
When a change is confined to an unknown region $\Omega$, the intensity $\lambda(u)$ differs between $\Omega$ and $\mathcal S\setminus\Omega$, rendering it an unknown mixture of intensities. Hence the integral in $p_i(x)$ is infeasible to evaluate without prior knowledge of $\Omega$.

To address this challenge, we develop a localized score approximation that captures conditional intensity dynamics using only local event history, rather than modeling the global score function $f_i$ implied by \eqref{eq:pdf-s}.
This localized construction offers two key advantages:
($i$) it focuses on the most relevant nearby events, thereby preserving the dominant contributors to the local intensity; and
($ii$) by restricting attention to a small neighborhood, it avoids incorporating distant regions that may lie under a different regime, yielding a locally stable approximation that is more likely to reflect a single underlying intensity.
See \Cref{sec:theory} for formal guarantees.

Specifically, let $\delta>0$ be a fixed spatial scale.
For any event $x$, define its $\ell_\infty$ spatial neighborhood as
\begin{equation}
\label{eq:ball}
\mathcal B_\delta(x)
\coloneqq
\{ s' \in \mathcal S : \|s' - s(x)\|_\infty \le \delta \}.
\end{equation}
We then define the conditional density of $x$ under regime $i$, given its local history, as
\begin{equation}
    \label{eq:cond-pdf-loc}
    p_i(x;\delta)
    =
    \lambda_i(x)
    \exp\left(
    -\int_{[t_n,t(x))\times\mathcal B_\delta(x)}
    \lambda_i(u)\,\mathrm du
    \right),
\end{equation}
and denote the corresponding score model by $f_i(x;\delta):=\nabla_x \log p_i(x;\delta)$.
A formal derivation of \eqref{eq:cond-pdf-loc} is provided by \Cref{lem:B-cond-pdf} in the appendix.
In our implementation, each score model $f_i(\cdot;\delta)$ is parameterized by a hybrid neural architecture consisting of a recurrent component (\eg, an RNN) and a feedforward component.This representation is then concatenated with covariates of the current event and input into a feedforward network, which outputs the corresponding score value. A schematic overview of the model architecture is provided in \Cref{app:exp}.
Each score model $f_i(\cdot;\delta)$ is learned via denoising score matching \cite{dong2025conditional, zhou2025sequential} using offline reference datasets $\mathcal D_0$ and $\mathcal D_1$ drawn from the pre- and post-change regimes, respectively.
Specifically, we minimize the denoising score matching loss
\begin{equation}
\label{eq:dsm}
L(x,f)
\coloneqq
\left\|
f(x + \epsilon;\delta) + \frac{\epsilon}{\sigma^2}
\right\|_2^2,
\quad
\epsilon \sim \mathcal N(0,\sigma^2 I),
\end{equation}
where $\epsilon$ are i.i.d. sampled noise data, and $\sigma>0$ is a user-specified noise level controlling the scale of perturbations.
The resulting score estimators are obtained as
\[
f_i
=
\arg\min_{f\in\mathcal F}
\sum_{x\in\mathcal D_i}
L(x,f),
\quad i=0,1.
\]
Extensions to online settings without access to post-change reference data are discussed in Remark~\ref{sec:extend}.

\subsection{Proposed Algorithm: \alg}
\label{sec:opt}

To operationalize the proposed detection statistic in \eqref{eq:stat}, we develop an alternating optimization algorithm, referred to as \alg, that solves the resulting two-stage maximization problem.
The core idea is to decompose the joint optimization over the change-point and spatial region into two coupled subproblems, which are solved iteratively over $K$ epochs.

At iteration $k$, given the current estimate of the change-point $\hat\tau^{(k-1)}$, we first update the spatial change region by solving
\begin{equation}
\label{eq:inner}
\hat \Omega^{(k)}
=
\arg\max_{\Omega \subseteq \mathcal S}
\int_{[\hat \tau^{(k-1)},\, t)\times\Omega}
\Delta(x)\,\mathrm d\mathbb N(x).
\end{equation}
Holding the updated region $\hat \Omega^{(k)}$ fixed, we then update the change-point estimate by solving
\begin{equation}
\label{eq:outer}
\hat \tau^{(k)}
=
\arg\max_{0 \le \tau \le t}
\int_{[\tau,\, t)\times\hat \Omega^{(k)}}
\Delta(x)\,\mathrm d\mathbb N(x).
\end{equation}

After $K$ iterations, the algorithm returns $(\hat\tau^{(K)}, \hat\Omega^{(K)})$ as the final estimates of the change-point and affected spatial region.
In what follows, we describe efficient procedures for solving \eqref{eq:inner} and \eqref{eq:outer}, which we refer to as the \emph{inner optimization step} (\textit{I-step}) and the \emph{outer optimization step} (\textit{O-step}), respectively.

\paragraph{I-step}
Given $\hat \tau^{(k - 1)}$ from the previous iteration, we first present the following theoretical insight.

\begin{figure}
    \centering
    \includegraphics[width=1.0\linewidth]{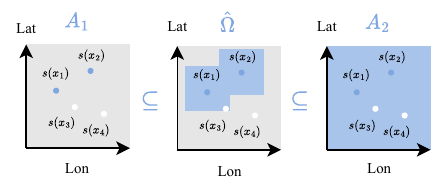}
    \vspace{-3ex}
    \caption{Illustrative example of \Cref{lem:level-set}.
    A dataset with four events $\{x_i\}_{i=1}^4$, the $x_1$ and $x_2$ have $\Delta(x_i) > 0$ while the $x_3$ and $x_4$ have $\Delta(x_i) < 0$. Their spatial coordinates are indicated by blue and white dots, respectively.
    }
    \label{fig:lemma-level-set}
\end{figure}

\begin{lemma}
    \label{lem:level-set}
    The solution to \eqref{eq:inner} is non-unique, and is given by the following family of sets:
    \[
         \hat{\Omega}^{(k)} \in \Big\{ \Omega \subseteq \mathcal{S} : A_1 \subseteq \Omega \subseteq A_2
         \Big\},
    \]
    where $A_1, A_2 \subset \mathcal{S}$ are defined as
    \begin{align*}
        A_1
        & =  \Big\{ s(x) : x \in \mathcal{H}_t \setminus \mathcal{H}_{\hat \tau^{(k-1)}}
        \ \text{and} \
        \Delta(x) > 0  \Big\}, \\
        A_2
        & =  \mathcal{S} \setminus \Big\{ s(x) : x \in \mathcal{H}_t \setminus \mathcal{H}_{\hat \tau^{(k-1)}}
        \ \text{and} \
        \Delta(x) < 0  \Big\},
    \end{align*}
\end{lemma}
As illustrated by Figure~\ref{fig:lemma-level-set}, $A_1$ is a finite collection of spatial coordinates whose corresponding anomaly measure is positive, and $A_2$ is the whole space $\mathcal{S}$ excluding a finite collection of spatial coordinates whose corresponding anomaly measure is negative.
Although \Cref{lem:level-set} states that the solution to \eqref{eq:inner} is not unique, in practice, its selection has important implications for numerical stability. To this end, we recommend the following construction:
\begin{equation}
    \label{eq:est-R}
    \hat \Omega^{(k)} \coloneqq
    \bigcup_{s \in A_1}
    \{ s' \in \mathcal{S} : \| s - s' \|_\infty \leq \delta \}
    \setminus
    A_2^c
    ,
\end{equation}
where $\delta$ is the radius of the spatial neighborhood for computing $\tilde x$.
Note that \eqref{eq:est-R} is efficient to compute, since either side of the minus sign are sets with finite elements\footnote{The left side is a finite union of balls, which is geometrically simple.}.

\paragraph{O-Step}
Given $\hat \Omega^{(k)}$, the outer problem \eqref{eq:outer} can be reduced to solving
\begin{equation}
    \label{eq:est-nu}
    \hat \tau^{(k)} = \underset{\tau \in \{ t \} \cup \{ t_j \}_{j = 1}^{N(t)} }{\arg\max} \int_{[\tau, t) \times \hat \Omega^{(k)}} \Delta(x) \d \mathbb{N}(x),
\end{equation}
since the counting measure only has jumps at a finite set of points.
This reduces the problem to a finite maximization problem, which is efficient to solve.
The complete algorithm is summarized in \Cref{alg:offline}.

\begin{remark}[Online Score Estimation]\label{sec:extend}
The proposed score estimation can be extended to settings where $\mathcal{D}_1$ is unavailable, requiring online estimation of the post-change score function. 
Specifically, we initialize both the pre- and post-change score models with the pre-change historical data $\mathcal D_0$. 
During the detection phase, for each incoming data $x_j$, along with the previously observed data stream, we estimate the change-point $\hat\tau_j^{(K)}$ and change region $\hat{\Omega}_j^{(K)}$ by running the alternating optimization algorithm (\Cref{sec:opt}) for $K$ epochs.
Then, the post-change model is updated via single-step (or multiple-step) gradient descent:
\begin{equation}
    \label{eq:grad-dec}
    \theta \leftarrow \theta - \eta \cdot \nabla_{\theta} \int_{[\hat \tau^{(K)}_j, t) \times \hat \Omega^{(K)}_j} L(x, f_1) \d \mathbb{N}(x),
\end{equation}
where $\theta$ denotes the parameters of the post-change score model $f_1$,
$\eta > 0$ is the learning rate hyperparameter,
and $L(\cdot)$ is the denoising score matching loss defined in \eqref{eq:dsm}.
Finally, the scoring rules $\psi_i(\cdot)$ and the anomaly measure $\Delta(\cdot)$ are then computed using the updated post-change model for the next incoming test data.
The pseudo code is provided in \Cref{alg:online} in the Appendix.
\end{remark}

\begin{algorithm}[t]
\caption{\alg{} Detection Procedure} 
\label{alg:offline}
\begin{algorithmic}[1]
    \REQUIRE Reference dataset  $\mathcal{D}_{0}, \mathcal{D}_{1}$;
    Test dataset stream $\mathcal{D}$;
    Optimization epochs $K$;
    Threshold $\gamma$.
    \STATE  Train score models $f_{0}$ and $f_{1}$.
    \FOR{$x_j \in \mathcal{D}$}
        \STATE Compute $\psi_0(x_j)$ and $\psi_1(x_j)$ as defined in \eqref{eq:hyvarinen}.
        \STATE $\Delta(x_j) \leftarrow \psi_0(x_j) - \psi_1(x_j)$.
    \ENDFOR
    \STATE Initialize $t \leftarrow 0$ and $j = 1$.
    \WHILE{$W_t < \gamma$}
    \STATE $t \leftarrow t(x_j)$, $j \leftarrow j + 1$
    \STATE Initialize $\hat \tau^{(0)} \leftarrow 0$.
    \FOR{$k \in \{ 1, \ldots, K \}$}
        \STATE $\hat \Omega^{(k)} \leftarrow$ Compute change region as in \eqref{eq:est-R}.
        \STATE $\hat \tau^{(k)} \leftarrow$ Compute change-point as in \eqref{eq:est-nu}.
    \ENDFOR
    \STATE Set $\hat \tau \leftarrow \hat \tau^{(K)}$ and $\hat \Omega \leftarrow \hat \Omega^{(K)}$.
    \STATE Compute spatio-temporal CUSUM statistic:
    $$
    W_t \leftarrow \int_{[\hat \tau, t) \times \hat \Omega} \Delta(x) \d \mathbb{N}(x).
    $$
    \ENDWHILE
    \STATE Set stopping time $\nu \leftarrow t$.
    \RETURN Estimation $(\hat \tau, \hat \Omega)$; Stopping time $\nu$.
\end{algorithmic}
\end{algorithm}

\section{Theoretical Results}
\label{sec:theory}

This section establishes theoretical guarantees for the proposed algorithm in terms of false alarm performance, detection delay, and change-region localization accuracy. 
For simplicity, we assume the learned score model to be exact, \ie, $f_i(x; \delta) = \nabla_x \log p_i(x;\delta)$. Note that this is possible when the score model class is sufficiently expressive and trained with ample data.
We denote $\mathbb{E}_\tau$ as the expectation under the distribution where the change happened at $\tau$.
All proofs are deferred to the appendix.

First, we analyze the false alarm performance of the proposed algorithm.
We adopt the average run length (ARL) metric, defined as the expected value of stopping time under the pre-change domain $\mathbb{E}_\infty[\nu]$. In the following \Cref{lem:far}, we establish that the average run length (ARL) is lower bounded by an exponential function of the detection threshold, up to a multiplicative constant.
Consequently, to satisfy a prescribed ARL lower-bound requirement, it suffices to set $\gamma \asymp \log(\mathrm{ARL \cdot 2^{\mu_0 \cdot T}})$.

\begin{lemma}[False alarm rate]
\label{lem:far}
    Consider the stopping rule $\nu$ defined in \eqref{eq:stopping-time}, then for any $\gamma > 0$,
    $$
    \mathbb{E}_\infty[\nu] \geq \min\left\{ O\left(\frac{e^{\gamma} }{\mathbb{E}_\infty[2^{N(T)}]}\right),
    T
    \right\}.
    $$
\end{lemma}
To quantify the detection efficiency of the proposed algorithm, we first derive the drift of the detection statistic, which characterizes the growth rate of the \alg{} statistic after the change.
Our theory shows that this term resembles a weighted and localized version of the Fisher divergence.

\begin{lemma}[Drift]
    \label{thm:drift}
    Define
    $$
    D_F(f_0 \Vert f_{1} )
    \coloneqq \mathbb{E}_{0}
    \left\|
    \sqrt{w(x)} \odot \Bigl[ f_0 (x; \delta) - f_{1} (x; \delta) \Bigr]
    \right\|_2^2,
    $$
    $$
    D_F(f_1 \Vert f_{0} )
    \coloneqq \mathbb{E}_{\infty}
    \left\|
    \sqrt{w(x)} \odot \Bigl[ f_1 (x; \delta) - f_{0} (x; \delta) \Bigr]
    \right\|_2^2,
    $$
    Then, under some assumptions\footnote{See \eqref{eq:bound-cond} of \Cref{lem:bound-cond} in the Appendix.} on $w(x)$, there is
    \begin{align*}
        \mathbb{E}_{\infty} \left[ \Delta(x) \right] 
        = - D_F(f_1 \Vert f_0),
        \quad 
        \mathbb{E}_{0} \left[ \Delta(x) \right] = D_F(f_0 \Vert f_1).
    \end{align*}
\end{lemma}

It can be observed from \Cref{thm:drift} that the drift of the statistics is negative in the pre-change regime, while it is positive in the post-change regime. This confirms that the proposed scoring rule is proper and can detect changes.
Additionally, \Cref{thm:drift} can be used to derive further explicit forms of the drift with assumptions over the weight $w(x)$ and the intensity function $\lambda(x)$. For example, under the Hawkes process assumption \eqref{eq:hawkes}, we have the following general lemma.

\begin{lemma}[Score differences for Hawkes process]
    \label{lem:param}
    Suppose $\lambda(x)$ is defined as in \eqref{eq:hawkes}, then
    $$
    f_1(x; \delta) - f_0(x; \delta) = (\mu_0 - \mu_1) \cdot \left[ \frac{\alpha \boldsymbol k(x)}{\lambda_1(x) \lambda_0(x)} +
    \begin{pmatrix}
        4\delta^2 \\
        \mathbf{0}_2
    \end{pmatrix}
    \right].
    $$
    where $\boldsymbol{k}(x)$ is a vector defined as:
    $$
    \boldsymbol k(x) = \int_{[0, t_n) \times \mathcal{B}(x)} \left[ \nabla_x \kappa(x, x') + \nabla_{x'} \kappa(x, x') \right] \d \mathbb{N}(x'),
    $$
\end{lemma}

\Cref{lem:param} further enables explicit characterization of the drift term under specific choices of the kernel function $\kappa(\cdot,\cdot)$.
In particular, when $\kappa$ is radial, the term $\boldsymbol{k}(x)$ vanishes identically, yielding a simplified expression for the drift.
We summarize the resulting forms in two representative cases:
$(i)$ a homogeneous Poisson process (HPP; \Cref{ex:hpp}) and
$(ii)$ a Hawkes process with an exponential kernel (\Cref{ex:hawkes}).

\begin{cor}
    \label{cor:kernel}
    Let $w(x) \coloneqq x$. We have the following result.
    \begin{enumerate}
        \item   
        If $\lambda(x)$ is a Poisson process (\Cref{ex:hpp}),
        \begin{align*}
            D_F(f_{1-i} \Vert f_{i})
            = \frac{4(\mu_0 - \mu_1)^2 \delta^2}{ \mu_i},
        \end{align*}
        \item If $\lambda(x)$ is a Hawkes process with exponential excitation kernel (\Cref{ex:hawkes}),
        \begin{align*}
            D_F(f_{1-i} \Vert f_{i})
            = \left[ \frac{4(\mu_0 - \mu_1)^2 \delta^2}{ \mu_i}\right](1 - \alpha).
        \end{align*}
    \end{enumerate}
\end{cor}

\Cref{cor:kernel} has two implications worth highlighting:
($i$) For HPP, the drift scales with the difference of the base rates $|\mu_0 - \mu_1|$ and the neighborhood radius $\delta$. 
This reflects that a sharper pre- and post-change intensity contrast and a larger region of observation show stronger evidence of change.
($ii$) For Hawkes processes with exponential kernels, the drift magnitude is discounted by a factor of $1- \alpha$ compared to the HPP setting. This reflects that a larger branching ratio $\alpha$, which corresponds to stronger self-excitation and increased temporal variability, renders evidence of change less pronounced, as observed deviations may be confounded by endogenous bursts or self-driven departures from stationarity.

Finally, we use \Cref{thm:drift} to derive the detection performance guarantees for the proposed algorithm. We examine two metrics: expected detection delay (EDD) \cite{li2017detecting, wang2023sequential} defined as 
$
\mathbb{E}_\tau[\nu - \tau \mid \nu > \tau ],
$
which evaluates the estimated change-point; and the Jaccard index of the estimated change region, defined as:
$
J(\hat \Omega, \Omega) = {| \hat \Omega \cap \Omega |} / { | \hat \Omega \cup \Omega |},
$
which evaluates the accuracy of the recovered change region $\hat \Omega$.


\begin{thm}
    \label{thm:performance}
    Under some regularity conditions (\Cref{app:performance}):
    \begin{enumerate}
        \item
        For a given threshold value $\gamma > 0$ and any true change-point $\tau \geq 0$, there is:
        \[
        \text{EDD} 
        \lesssim
        \min\left\{
        \frac{\gamma|\mathcal{S}|}{
        \bar \lambda_1 |\Omega|
        D_F(f_0 \Vert f_1)}(1 + o(1)), \
        T
        \right\}
        \]
        where $\bar \lambda_1$ is defined as the stationary intensity of the post-change process. 
        \item
        The Jaccard index is asymptotically lower bounded:
        \[
            \lim_{\gamma \to \infty} J(\hat \Omega, \Omega) \geq \min \left\{ \frac{|\Omega|}{|\Omega \oplus \mathcal{B}_\delta|}, \right. \left. \frac{|\Omega \ominus \mathcal{B}_\delta|}{|\Omega|} \right\},
        \]
        where $\oplus$ and $\ominus$ denote the set dilation and erosion operators (\eqref{eq:dilation} and \eqref{eq:erosion}), and $\mathcal{B}_\delta$ is defined in \eqref{eq:ball}.
    \end{enumerate}
\end{thm}

The first statement of \Cref{thm:performance} provides an asymptotic upper bound on the EDD.
Unlike standard change-point detection results, this bound is negatively proportional to both the stationary post-change process intensity $\bar \lambda_1$ and the size of the change region.
This dependence arises because the statistic $W_t$ accumulates information through the underlying counting process, so the detection delay is governed by the rate at which post-change events are observed. 
The second statement of \Cref{thm:performance} shows that when allowing a sufficient amount of time for the statistic to accumulate evidence, the estimated change region would approximate the true change region by a difference of the localization neighborhood $\mathcal{B}_\delta$.
This gap arises because a localized region that crosses the boundary of $\Omega$ contains mixed pre-change and post-change processes, which disrupts the application of \Cref{thm:drift} in establishing clear convergence.

\section{Numerical Experiments}
\label{sec:exp}

We evaluate the performance of our proposed algorithm \alg{} on synthetic data and two real-world datasets: earthquake catalog dataset \cite{jma} and wildfire incident dataset \cite{calfire}. 
In all experiments, \alg{} is implemented with the score model parameterized by: 
($i$) a one-layer LSTM network with $32$ hidden units, which encodes event history, followed by
($ii$) a one-layer feed-forward network with $512$ hidden units, which takes as input the concatenation of the history embedding and the current observed event, outputting the score.
The detection threshold is determined by simulation to control the average run length \cite{wang2023sequential, zhou2025sequential}. For additional details, see \Cref{app:exp}.

\subsection{Simulation Results}

Synthetic event sequences are generated from a spatio-temporal Hawkes process \eqref{eq:hawkes} defined on the time range $[0, 1)$ and a unit two-dimensional spatial box $[0, 1] \times [0, 1]$.
The excitation kernel $\kappa(x,x')$ is specified as an exponential kernel function with $\beta = 0.1$ by default.
We set the true change point as $\tau = 0.5$ and the true change region as $\Omega = [0.4, 0.6] \times [0.4, 0.6]$.
The pre- and post-change base rates are set to $\mu_0 = 1 \times 10^2$ and $\mu_1 = 10 \times \mu_1$.
This configuration allows a moderate level of spatio-temporal spillover effect, while ensuring the detectability of the change.
The training data for $f_0$ and $f_1$ consists of a single stream of data generated from each of the pre- and post-change processes, respectively.

\subsubsection{Comparison with Baselines}
We compare the detection performance of \alg{} with five baselines:
($i$) CUSUM: the standard CUSUM statistic \cite{page-biometrica-1954} discussed in \Cref{sec:prelim};
($ii$) SCUSUM: the score-based CUSUM statistic introduced in \cite{wu2024quickest};
($iii$) PP CUSUM: the CUSUM statistics using an estimated point-process log-likelihood ratio as the anomaly measure defined in \eqref{eq:cusum-continuous};
($iv$) MinCUSUM($n$): the CUSUM statistics for multi-stream change detection introduced in \cite{mei2010efficient} using a $n \times n$ uniform spatial discretization as its channels.
All baselines, except PP-CUSUM, operate on preprocessed binned count data derived from the event sequences, as they cannot directly handle spatio-temporal event data.
MinCUSUM further requires a manual partition of the data into multiple channels.
Note that SCUSUM and PP CUSUM are omitted from comparing the Jaccard index, since their estimated change region is by default the entire spatial region, and would collapse to the result given by CUSUM.
These baselines are included to reflect the typical approaches used to address change-point detection under the problem setting described in \Cref{sec:setting}.

\begin{figure*}[!t]
    \centering
    \includegraphics[width=1.0\linewidth]{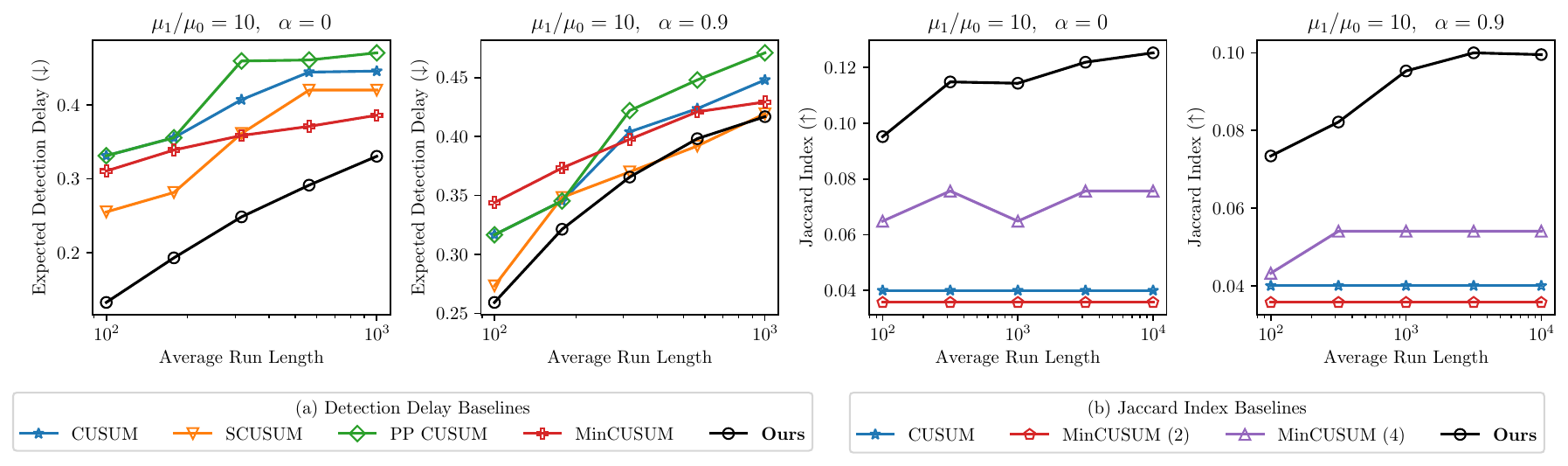}
    \vspace{-4ex}
    \caption{Detection performance versus average run length (ARL) under different settings. The first two figures are the expected detection delay versus ARL plots; the last two plots are the Jaccard index versus ARL plots. Within these sets of figures, we vary the branching ratio $\alpha$ between $0$ and $0.9$.
    }
    \label{fig:comparison}
\end{figure*}

\begin{figure}[!t]
    \centering
    \includegraphics[width=\linewidth]{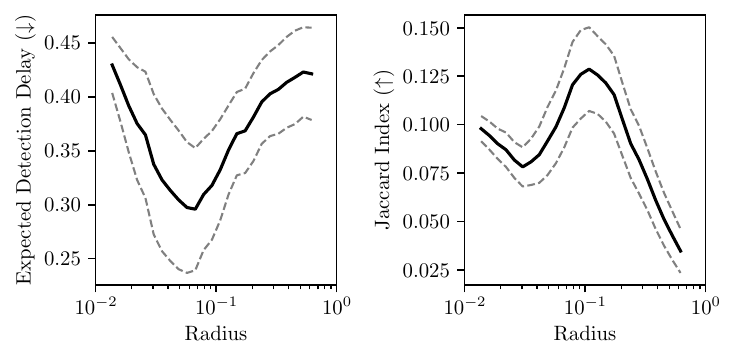}
    \vspace{-.2in}
    \caption{
    Model performance for varying neighborhood radii ($\delta$). Solid lines represent means; dashed lines represent the standard error of the mean.
    }
    \label{fig:radius}
\end{figure}

\Cref{fig:comparison} presents the performance metric versus average run length (ARL) tradeoff plot of \alg{} comparing with the above baselines.
The reported results are averaged over $100$ trials.
Across all ARL values and hyperparameter settings, \alg{} consistently shows a lower expected detection delay (EDD) and a higher Jaccard index compared to the other baselines.
This illustrates the superiority of the proposed algorithm, which benefits from a mix of the detection mechanism, statistics design, and optimization procedure we introduced earlier.
On a side note, we observe that all models achieve consistently higher EDD and lower Jaccard index under strong spatio-temporal spillover ($\alpha = 0.9$) compared to the no-spillover setting ($\alpha = 0$). This observation is consistent with \Cref{cor:kernel}, serving as additional evidence suggesting that spillover negatively impacts detection performance.
Beyond its superior performance, \alg{} operates directly on raw event sequences without the need for preprocessing, unlike most baseline methods, further highlighting its suitability for spatio-temporal change detection tasks.

\begin{figure*}
    \centering
    \includegraphics[width=\linewidth]{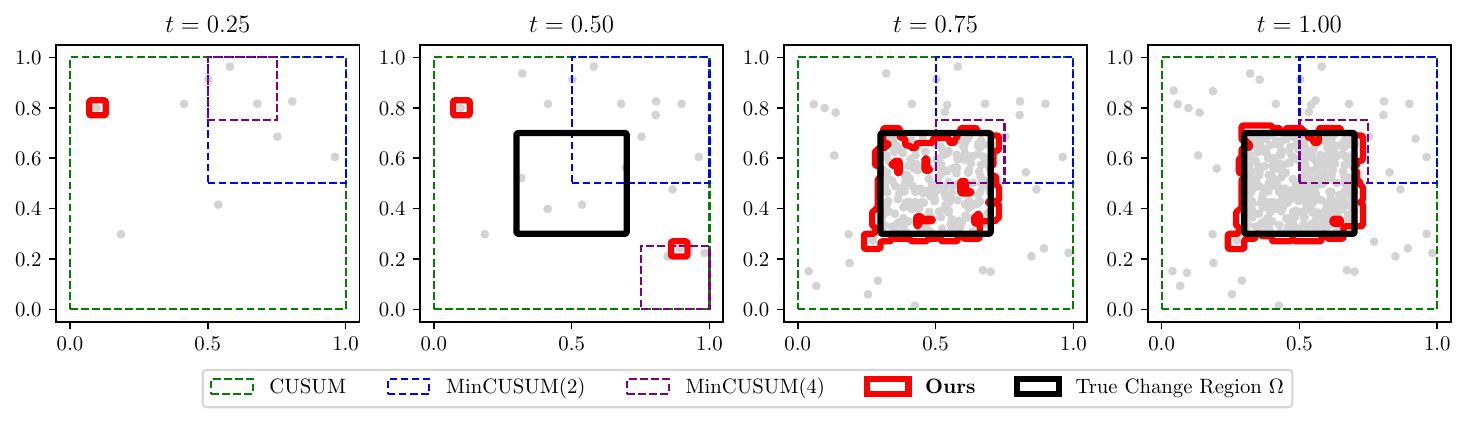}
    \includegraphics[width=\linewidth]{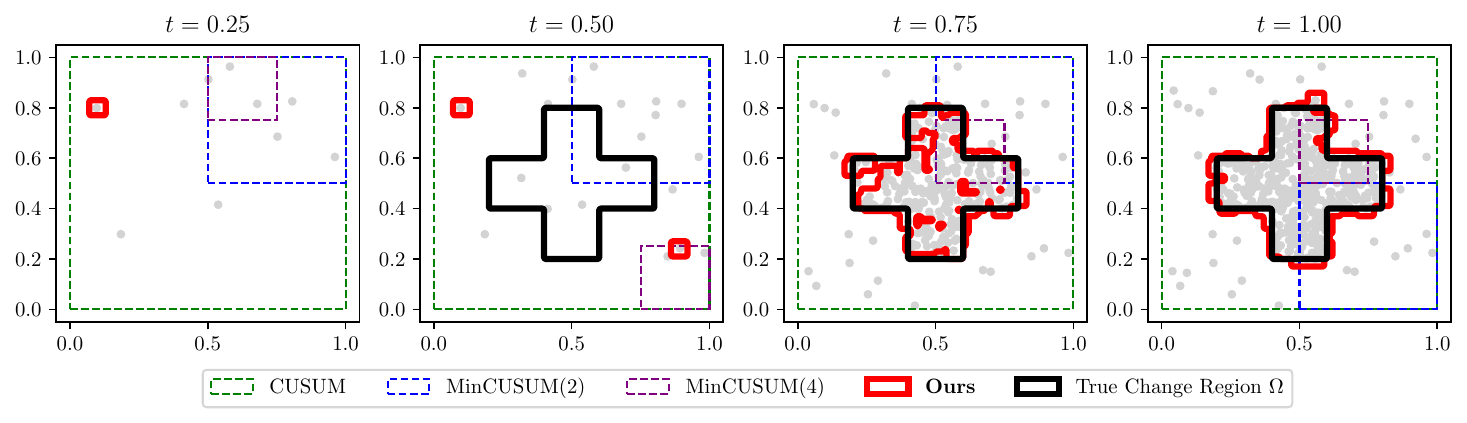}
    \includegraphics[width=\linewidth]{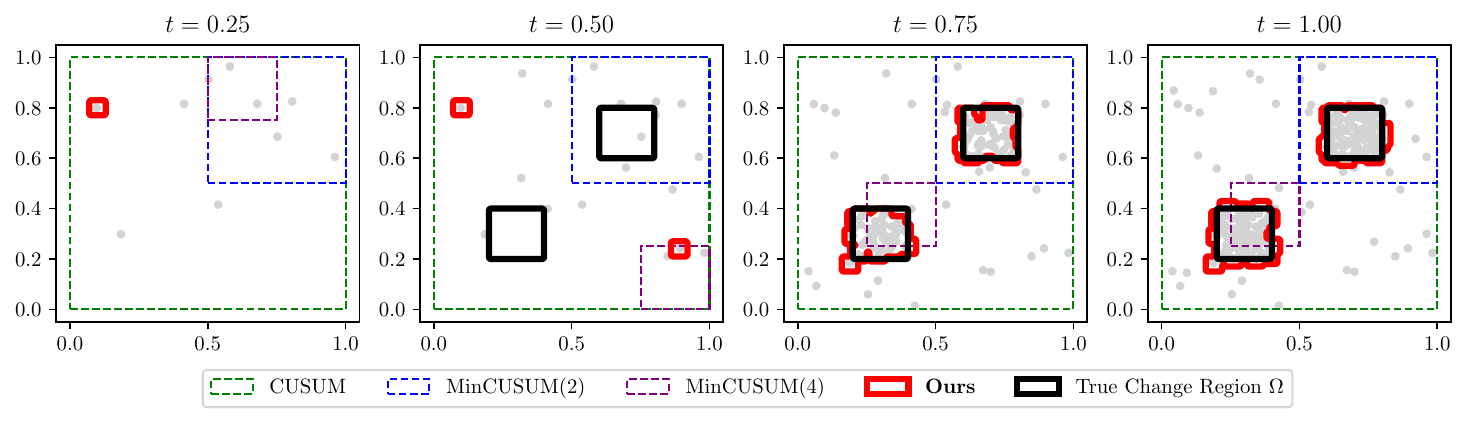}
    \caption{Temporal evolution of the estimated change region across different baselines, with each snapshot corresponding to treating the current time as the stopping time. We assume the true change point $\nu = 0.5$. The first row follows the existing setting, while the last two rows modify the true change region to more complex shapes (\textit{cross} and \textit{bimodal}). Gray dots represent event spatial coordinates.
    }
    \label{fig:spatial-evo}
\end{figure*}

Furthermore, \Cref{fig:spatial-evo} illustrates the temporal evolution of the estimated change regions under different synthetic experimental settings. CUSUM consistently estimates the change region as the entire spatial domain, which is expected since CUSUM (along with many other standard change point detection methods) is designed solely to raise temporal alarms and therefore interprets detected changes as global rather than spatially localized.
The two MinCUSUM variants are able to target more meaningful affected areas compared to CUSUM, but their localization remains coarse, as the estimated change regions are rectangular and are misaligned with the true change region $\Omega$. This limitation arises from the reliance on manual spatial discretization when casting spatio-temporal change point detection problems as multi-stream change point detection problems, which can lead to substantial localization errors whenever prior information of the change region geometry is unavailable. In contrast, across all settings, our method \alg{} accurately traces the contour of the true change region after the change occurs, even in challenging scenarios where the affected region is irregular or disconnected, without having access to prior information about the true change region shape.
This explains its superior Jaccard index performance in \Cref{fig:comparison} and highlights the unique applicability of \alg{} in settings where change regions are continuous, complex, and not well captured by predefined spatial partitions.

\subsubsection{Ablation Study on $\delta$}

We conduct an ablation study on the effect of the localized radius $\delta$ on the detection performance, as shown in \Cref{fig:radius}.
The detection-delay curve exhibits a rough convex trend, whereas the Jaccard-index curve displays a rough concave pattern. 
This reflects a fundamental tradeoff in choosing $\delta$: smaller $\delta$ yields neighborhoods that are too sparse to provide informative signals, while larger neighborhoods may contaminate the score function and detection statistics by mixing in events from both pre- and post-change regimes, both degrading detection performance.
Under the current synthetic setting, from \Cref{fig:radius}, a balanced choice of radius is $\delta = 0.1$, and is used across our synthetic experiment.

\subsubsection{Ablation Study on Iteration Number and Base Rate}

We also conduct ablation studies on the effect of different optimization rounds $K$ and pre-change base rate $\mu_0$ on the detection performance.
The pre-change base rate serves as a surrogate for the expected number of test sample input to the algorithm. 
\Cref{tab:param} presents the performance metrics evaluated on \alg{} under varying optimization iteration numbers $K=1,3,5$ and base rates $\mu_0 = 10^1,10^2,10^3$.
Across different values of $\mu_0$, both EDD and the Jaccard index remain nearly unchanged, indicating that although the proposed statistic updates only at discrete event times, its detection performance is robust to event frequency as opposed to our observation for the EDD upper bound derived in \Cref{thm:performance}.
Meanwhile, increasing $K$ consistently improves both EDD and the Jaccard index, suggesting that the optimization procedure benefits from allowing sufficient iterations for convergence.
Finally, although increasing either $K$ or $\mu_0$ leads to higher computational cost, the runtime never exceeds one second per trial, demonstrating the method’s computational feasibility for large-scale datasets. This efficiency arises from the score-based design of the statistic, which avoids costly numerical integration.

\begin{table}[!t]
\centering
\caption{\alg{} performance under different $K$ and $\mu_0$ specifications.}
\begin{adjustbox}{max width=\linewidth}
\begin{tabular}{cc|ccc}
\toprule
 &  & EDD & Jaccard Ind. & Run time (s) \\
\midrule
\midrule
\multirow[c]{3}{*}{$K=1$} & $\mu_0=10^1$ & $0.39 \pm 0.05$ & $0.04 \pm 0.01$ & $0.01 \pm 0.00$ \\
 & $\mu_0=10^2$ & $0.39 \pm 0.05$ & $0.04 \pm 0.01$ & $0.03 \pm 0.00$ \\
 & $\mu_0=10^3$ & $0.39 \pm 0.05$ & $0.04 \pm 0.01$ & $0.06 \pm 0.00$ \\
\midrule
\multirow[c]{3}{*}{$K=3$} & $\mu_0=10^1$ & $0.38 \pm 0.04$ & $0.11 \pm 0.02$ & $0.03 \pm 0.00$ \\
 & $\mu_0=10^2$ & $0.38 \pm 0.04$ & $0.11 \pm 0.02$ & $0.08 \pm 0.01$ \\
 & $\mu_0=10^3$ & $0.38 \pm 0.04$ & $0.11 \pm 0.02$ & $0.14 \pm 0.00$ \\
\midrule
\multirow[c]{3}{*}{$K=5$} & $\mu_0=10^1$ & $0.15 \pm 0.04$ & $0.09 \pm 0.01$ & $0.13 \pm 0.01$ \\
 & $\mu_0=10^2$ & $0.15 \pm 0.04$ & $0.09 \pm 0.01$ & $0.29 \pm 0.01$ \\
 & $\mu_0=10^3$ & $0.15 \pm 0.04$ & $0.09 \pm 0.01$ & $0.45 \pm 0.01$ \\
\midrule
\bottomrule
\end{tabular}
\end{adjustbox}
\label{tab:param}
\end{table}

\subsection{Real Data Experiments}

\begin{figure}[!t]
    \centering
    \includegraphics[width=1.0\linewidth]{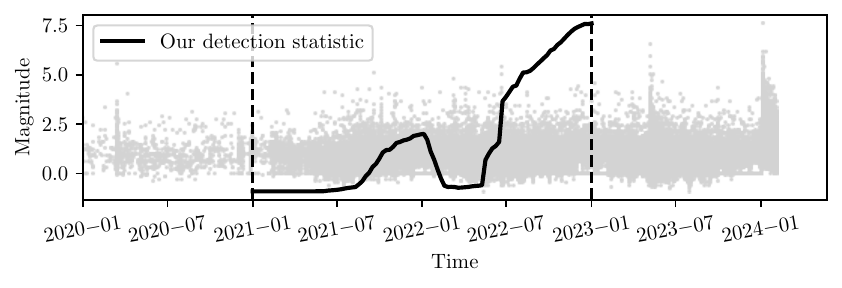}
    \includegraphics[width=1.0\linewidth]{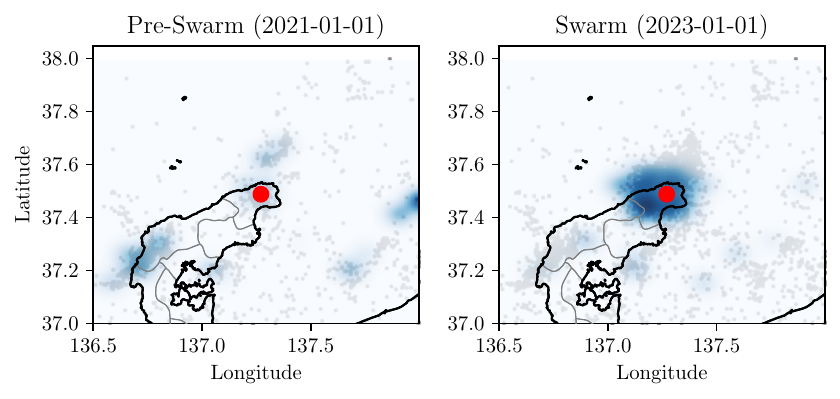}
    \vspace{-0.3in}
    \caption{\alg{} implementation outcome on the earthquake dataset. Gray dots represent event points.
    \textit{Top figure:} Black line represents the normalized detection statistic.
    \textit{Two bottom figures:} Red dot indicates the major shock epicenter, blue shaded regions represent the averaged estimated change regions.}
    \label{fig:earthquake}
    \vspace{-0.1in}
\end{figure}

We apply \alg{} to two real-world datasets to demonstrate its practical utility. In these experiments, we do not assume the availability of historical data from the post-change regime and adopt the corresponding extension of our algorithm, as detailed in \Cref{sec:extend}, to sequentially estimate and update the post-change score model using streaming data.

\subsubsection{Earthquake Early Warning}
We use the Japan Meteorological Agency (JMA) earthquake catalog for the Noto region of Japan \cite{jma}, covering the period from 2018 to 2024. During this interval, the region experienced two major earthquakes: the 2023 $M_j$ 6.5 event and the 2024 $M_j$ 7.6 event. Seismological studies have identified the years 2021–2023 as a swarm period \cite{peng2025evolution}, characterized by an anomalous increase in seismic activity that may signal elevated earthquake risk. Our objective is to automate the identification of the swarm periods using \alg{} to support earthquake early warning.

\Cref{fig:earthquake} (top) visualizes the resulting \alg{} detection statistic, which exhibits a significantly increasing trend beginning in late 2021, followed by another sharp rise in mid-2022. These time points potentially indicate two time marks for seismic activities underlying the subsequent major earthquakes, and indicate that \alg{} is able to capture change signals during the swarm period. \Cref{fig:earthquake} (bottom) presents two representative snapshots of the estimated change region, corresponding to the pre-change and post-change regimes, respectively. 
It can be seen that before entering the swarm period, \alg{} was unable to point to the main shock areas. However, after the swarm period, \alg{} clearly identifies the main shock region.
These results show that \alg{} can not only be deployed as an effective earthquake early warning and localization tool, but the findings also serve as valid evidence supporting swarm theory on earthquake precursors in seismology research.

\begin{figure}[t!]
    \centering
    \includegraphics[width=1.0\linewidth]{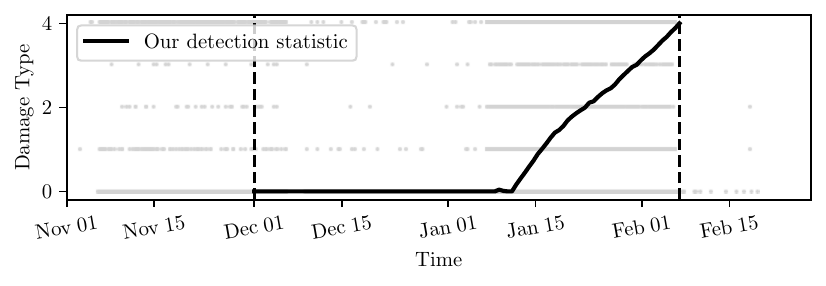}
    \includegraphics[width=1.0\linewidth]{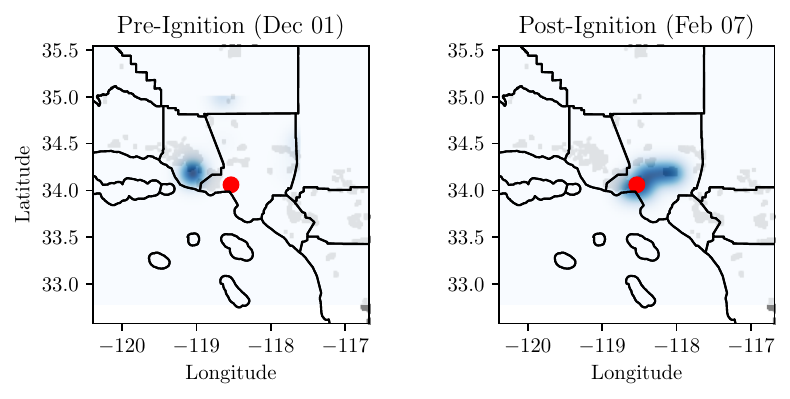}
    \vspace{-0.3in}
    \caption{\alg{} implementation outcome on the wildfire dataset. Gray dots represent event points.
    \textit{Top figure:} Black line represents the normalized detection statistic.
    \textit{Two bottom figures:} Red dot indicates the Palisades fire epicenter, blue shaded regions represent the averaged estimated change regions.}
    \label{fig:wildfire}
    \vspace{-0.2in}
\end{figure}

\subsubsection{Wildfire Detection}

We use the damage inspection dataset publicly released by the California Department of Forestry and Fire Protection \cite{calfire}. The dataset records structures affected by wildfires within a 100-meter perimeter and includes attributes such as event time, geographic coordinates, and damage severity. We focus on the Palisades wildfire, which ignited on January 7, 2025, in the Los Angeles region and caused substantial societal disruption due to its extensive destruction and negative social impact. Similar to the previous earthquake setting, our objective is to simulate an automated early-detection scenario for the Palisades wildfire using this dataset.

\Cref{fig:wildfire} presents the results of the wildfire case study. Although the detection statistic exhibits a clear upward drift toward the end of the monitoring period, the change point occurs precisely on the day the wildfire ignited rather than earlier. This suggests that, unlike the earthquake setting where the swarm period served as a meaningful precursor, no detectable precursor signals were identified by \alg{} in the wildfire scenario. This behavior is consistent with the nature of wildfire ignition, which often occurs abruptly or is even human-initiated, leaving little to no early warning in the event data.
While in the pre-ignition period, the observed wildfire hotspots originate from a separate, unrelated wildfire event, \alg{} correctly traces the estimated change region after ignition.
This result indicates that \alg{} can be deployed to provide rapid post-ignition localization and warning following the onset of the wildfire.

\section{Conclusion}

In this paper, we propose a score-based CUSUM statistic for spatio-temporal change detection in spatio-temporal point process settings, where the change is assumed to occur at an unknown continuous time and over a continuous spatial region. By combining denoising score matching, a weighted Hyvärinen score, and spatial localization, we construct a well-defined detection statistic that avoids explicit integration of the intensity function, yielding substantial computational advantages over traditional likelihood-based approaches.
To efficiently compute the statistic, we develop an alternative optimization scheme for the resulting nested maximization problem, enabling scalable optimization over continuous time and space. On the theoretical side, we establish performance guarantees for the false-alarm rate, characterize parametric forms of the drift, derive bounds on the expected detection delay, and analyze the accuracy of the estimated change region via the Jaccard index.
Through extensive synthetic experiments and two real-world case studies, we demonstrate the superiority of the proposed method over existing baselines for spatio-temporal change detection and provide practical insights into its behavior and deployment.

Several promising directions can be explored in future research.
($i$) Our current theoretical analysis characterizes the parametric form of the drift under the assumptions of either no kernel excitation or an exponential kernel. Dedicated theoretical studies could extend this analysis to broader classes of excitation kernels, yielding a more complete understanding of the kernel type's impact on the statistic's detection performance.
($ii$) Empirically, we observe that the performance of the proposed algorithm depends on the localization radius $\delta$ and the number of optimization iterations $K$. Future work may develop a principled way of configuring these hyperparameters, possibly by investigating their theoretical link with the detection performance, or develop a calibration algorithm that automates their selection in practical implementations.

\bibliographystyle{ieeetr}
\bibliography{ref}

\clearpage
\newpage
\appendix

Throughout the proofs, we fix $\delta$ and suppress the dependence on $\delta$ for notational simplicity.

\subsection{Proof of \Cref{lem:level-set}}
\label{app:level-set}

\begin{proof}
    We begin by defining the auxiliary function:
    $$
    L(s) \coloneqq 
    \int_{[\hat \tau^{(k - 1)}, T)} \Delta(x) \d \mathbb{N}(t)
    $$
    where $\mathbb{N}(t) \coloneqq \mathbb{N}([0, T) \times \mathcal{S})$ denotes the marginalized counting process of the temporal dimension. Then, \eqref{eq:outer} can be rewritten:
    \begin{equation}
        \label{eq:L}
        \hat \Omega^{(k)} = \arg\max_{\Omega \subseteq \mathcal{S}} \int_{\Omega} L(s) \d \mathbb{N}(t).
    \end{equation}
    The integral is maximized when the set is mapped to all positive (and nonzero) integrands, 
    \begin{equation}
        \label{eq:sandwich}
         \hat \Omega \in \Big\{ \Omega \subseteq \mathcal{S} : L^{-1}(\mathbb{R}_+) \subseteq \Omega \subseteq L^{-1}(\mathbb{R}_+ \cup \{0\})
         \Big\},
    \end{equation}
    where $L^{-1}(\cdot)$ is defined as the preimage set of $L$, \ie, for any $\mathcal{I} \subseteq \mathbb{R}$, there is
    $$
    L^{-1}(\mathcal{I}) \coloneqq \{ s \in \mathcal{S} : L(s) \in \mathcal{I} \}.
    $$
    By the definition of $L(\cdot)$, there is 
    \begin{align*}
        L^{-1}(\mathbb{R}_+)
        & =  \Big\{ s(x) : x \in \mathcal{D}'
        \ \text{and} \
        \Delta(x) > 0  \Big\}, \\
        L^{-1}(\mathbb{R}_+ \cup \{0\})
        & =  \mathcal{S} \setminus \Big\{ s(x) : x \in \mathcal{D}'
        \ \text{and} \
        \Delta(x) < 0  \Big\},
    \end{align*}
    where we denote $\mathcal{D}' = \mathcal{H}_t \setminus \mathcal{H}_{\hat \tau^{(k-1)}}$ for notation simplicity.
    Note that the two sets of spatial coordinates satisfying $\Delta(x) < 0$ and $\Delta(x) > 0$ are all finite sets since $\mathcal{D}$ is finite. Therefore, we conclude the proof.
\end{proof}

\subsection{Proof of \Cref{lem:far}}
\label{app:far}

We first state two lemmas, which prove the existence of useful proxy processes with nice theoretical properties. They will be used as an intermediate step in the derivation of the bound. 

\begin{lemma}
    \label{lem:proxy}
    Denote the following process
    \begin{equation*}
        \tilde W_{u:t}^\Omega \coloneqq
        \int_{[u, t) \times \Omega} \left\{ \zeta \cdot \Delta(x) - \log \left( \mathbb{E}_\infty[e^{\zeta \cdot \Delta(x)}]\right) \right\} \d \mathbb{N}(x).
    \end{equation*}
    Then, for any $\zeta > 0$, $u < t$ and $\Omega \subseteq \mathcal{S}$, $\exp(\tilde W_{t}^\Omega)$ is a nonnegative $\mathbb{P}_\infty$-martingale with mean 1.
\end{lemma}

\begin{proof}[Proof for \Cref{lem:proxy}]
    We abbreviate $\tilde W_{u:t}^\Omega$ as $\tilde W_{t}$ when clear from context.
    For any $\zeta > 0$, $u < t$ and $\Omega \subseteq \mathcal{S}$,
    it is easy to see non-negativity. To prove that it is a $\mathbb{P}_\infty$-martingale, notice that for any $\tau \in [0, T)$,
    \allowdisplaybreaks
    \begin{align*}
        & \mathbb{E}_\infty[e^{\tilde W_t} | \mathcal{H}_{\tau}] = \mathbb{E}_\infty[ e^{\tilde W_{t'}} | \mathcal{H}_{\tau}] \\
        &  \times \mathbb{E}_\infty\left[ \exp\left(  \int_{[t', t) \times \Omega} \zeta \cdot \Delta(x) \d \mathbb{N}(x) \right) | \mathcal{H}_{\tau} \right] \\
        &  \div \mathbb{E}_\infty \left[
        \exp\left( \int_{[t', t) \times \Omega} \log \left( \mathbb{E}_\infty[e^{\zeta \cdot \Delta(x)}]\right) \d \mathbb{N}(x)
        \right) | \mathcal{H}_{\tau} \right] \\
        & = \mathbb{E}_\infty[ e^{\tilde W_{t'}} | \mathcal{H}_{\tau}] \times \mathbb{E}_\infty[e^{\zeta \cdot \Delta(x)} ] \div \mathbb{E}_\infty[e^{\zeta \cdot \Delta(x
        )} ] \\
        & =\mathbb{E}_\infty[ e^{\tilde W_{t'}} | \mathcal{H}_{\tau}],
    \end{align*}
    where $t'$ denotes the last event timestamp prior to $t$.
    We can then recurse on $\mathbb{E}_\infty[ e^{\tilde W_{t'}} | \mathcal{H}_{\tau}]$ by unravelling the last event once at a time to derive
    $$
    \mathbb{E}_\infty[e^{\tilde W_t} | \mathcal{H}_{\tau}] = 
    \mathbb{E}_\infty[e^{\tilde W_{t'}} | \mathcal{H}_{\tau}] =
    \ldots = 
    e^{\tilde W_\tau}.
    $$
    This proves that it is a $\mathbb{P}_\infty$-martingale.
    To derive its mean,
    simply note that
    $$
    \mathbb{E}_\infty[e^{\tilde W_t}] = \mathbb{E}_\infty[e^{\tilde W_u}] = e^0 = 1
    $$
    We finish the proof.
\end{proof}

\begin{lemma}
    \label{lem:proxy-2}
    Let $\tilde W_{u:t}^\Omega$ be defined as in \Cref{lem:proxy}. Consider the following process:
    $$
    M_t^\Omega = \int_0^t \left\{ \exp(\tilde W_{u:t}^\Omega) -  1 \right\} \d u.
    $$
    We abbreviate $M_t^\Omega$ as $M_t$ when clear from context. 
    Then for any $\Omega$, $M_t$ is a $\mathbb{P}_\infty$-martingale with mean $0$.
\end{lemma}

\begin{proof}
    Note that for any $t' < t$, there is the following recursion:
    \begin{align*}
        M_t
        & = \int_0^t \left\{ \exp(\tilde W_{u:t}^\Omega) -  1 \right\} \d u \\
        & = \int_0^{t'} \exp(\tilde W_{u:t}^\Omega) \d u + \int_{t'}^{t} \exp(\tilde W_{u:t}^\Omega) \d u - t. 
    \end{align*}
    The first term can be further decomposed as
    \begin{align*}
        \int_0^{t'} \exp(\tilde W_{u:t}^\Omega) \d u
        & = \int_0^{t'} \exp(\tilde W_{u:t'}^\Omega) \cdot \exp(\tilde W_{t':t}^\Omega) \d u \\
        & = M_{t'} \cdot \exp(\tilde W_{t':t}^\Omega).
    \end{align*}
    Consequently, there is
    \allowdisplaybreaks
    \begin{align*}
        \mathbb{E}_\infty[M_t \mid \mathcal{H}_{t'} ]
        & = \int_0^{t'} \exp(\tilde W_{u:t'}^\Omega) \d u \cdot \mathbb{E}_\infty[\exp(\tilde W_{t':t}^\Omega) \mid \mathcal{H}_{t'} ] \\
        & \quad + \int_{t'}^t \mathbb{E}_\infty [ \exp(\tilde W_{u:t}^\Omega)\mid \mathcal{H}_{t'} ] \d u - t \\
        & = \int_0^{t'} \exp(\tilde W_{u:t'}^\Omega) \d u - t' = M_{t'},
    \end{align*}
    where the second equality follows from \Cref{lem:proxy}. Therefore, we've proved that $M_t$ is a $\mathbb{P}_\infty$-martingale. To derive its mean, simply note that
    $$
    \mathbb{E}_\infty[M_t] = \mathbb{E}_\infty[M_0] = 0.
    $$
    We finish the proof.
\end{proof}

Now, we present the complete proof for \Cref{lem:far}.
The proof is adapted from the proof of Theorem V.2 of \cite{chen2025score}.

\begin{proof}
    Given prior distribution $\pi(\Omega)$ over spatial regions $\Omega \subseteq \mathcal{S}$, note the following relationship:
    \begin{align*}
        W_t
        & = \sup_{\Omega \subseteq \mathcal{S}} \sup_{0 \leq u < t} \int_{[u, t) \times\Omega} \Delta(x) \d \mathbb{N} (x) \\
        & \hspace{-1ex} \leq \sup_{\Omega \subseteq \mathcal{S}} \sup_{0 \leq u < t} \int_{[u, t) \times \Omega} \Delta(x) - \frac{1}{\zeta}\log \left( \mathbb{E}_\infty[e^{\zeta \cdot \Delta(x)}]\right)  \d \mathbb{N}(x) \\
        & = \sup_{\Omega \subseteq \mathcal{S}} \sup_{0 \leq u < t} \frac{1}{\zeta} \tilde W_{u:t}^\Omega 
        \leq \sum_{\Omega \in 2^\mathcal{S}}\int_{[0, t)} \frac{1}{\zeta} \tilde W_{u:t}^\Omega \d u.
    \end{align*}
    Denote the Voronoi cell associated with the $i$-th event as
    $$
    \mathcal{V}_i = \left\{
    s \in \mathcal{S} : \| s - s(x_i) \|_2 \leq \| s - s(x_j) \|_2, \quad \forall j \neq i  
    \right\}.
    $$
    We then define the collection of all Voronoi regions that can be formed by events observed up to time~$t$ as
    $$
    \mathcal{V}(t) = \left\{ \bigcup_{i \in \mathcal{I}} \mathcal{V}_i : \mathcal{I} \subseteq \{1, \ldots, N(t) \} \right\}.
    $$
    Then, \Cref{lem:level-set} implies the following equivalence:
    $$
    \sum_{\Omega \in 2^\mathcal{S}}\int_{[0, t)} \frac{1}{\zeta} \tilde W_{u:t}^\Omega \d u = \sum_{\Omega \in \mathcal{V}(t)}\int_{[0, t)} \frac{1}{\zeta} \tilde W_{u:t}^\Omega \d u
    $$
    
    Define the stopping time
    $$
    \tilde \nu \coloneqq \inf \left\{ t \in [0,T) :
    \sum_{\Omega \in \mathcal{V}(t)} \int_{[0, t)} \exp(\tilde W_{u:t}^\Omega) \d u  \geq e^{\gamma \zeta} \right\},
    $$
    then there must be $\nu \geq \tilde \nu$, so we only need to lower bound $\mathbb{E}_\infty[\tilde \nu]$.
    Consider the following process:
    $$
    \tilde M_t \coloneqq  \sum_{\Omega \in \mathcal{V}(t)} \underbrace{\int_{0}^t  \left\{ \exp(\tilde W_{u:t}^\Omega) - 1\right\} \d u}_{= M_t}
    $$
    By \Cref{lem:proxy-2}, we know that $M_t$ and $\tilde M_t$ are both $\mathbb{P}_\infty$ martingales with mean 0.
    Using the optional sampling theorem on $\tilde M_t$, there is
    \begin{align*}
        0
        = \mathbb{E}_\infty \left[ \tilde M_{0}\right] =  \mathbb{E}_\infty \left[ \tilde M_{\tilde \nu}\right]  \geq e^{\gamma\zeta} - \mathbb{E}_\infty \left[\tilde  \nu \cdot 2^{N(T)} \right],
    \end{align*}
    where the last inequality follows from the definition of $\tilde \nu$. Therefore, rearranging terms, we get
    $$
    \mathbb{E}_\infty[\nu] \leq  \mathbb{E}_\infty[\tilde  \nu] = \frac{e^{\gamma\zeta} }{\mathbb{E}_\infty[2^{N(T)}]} = O\left( \frac{e^{\gamma} }{\mathbb{E}_\infty[2^{N(T)}]}
    \right).
    $$
    where we suppress the dependency on the constants $\zeta$ and $C$.
    We complete the proof by noting that the ARL by definition must not be larger than the time horizon $T$.
\end{proof}

\subsection{Proof of \Cref{thm:drift}}
\label{app:drift}

To prove \Cref{thm:drift}, we first state a lemma adapted from \cite{liu2022estimating}.
The two conditions \eqref{eq:bound-cond} require that the weighting function must smooth out the PDFs at the boundary terms.
Note that this does not require that the PDFs themselves be vanishing at their boundaries, which makes it weaker than the general assumption for integrate-by-parts \cite{hyvarinen2005estimation}.

\begin{lemma}[Theorem 2 of \cite{liu2022estimating}]
    \label{lem:bound-cond}
    For any two continuous PDFs $p(x)$ and $q(x)$ supported on $[0, +\infty) \times \otimes_{k = 2}^{d-1} [a_k, b_k]$, suppose that the weighting function $\w(x)$ satisfies
    \begin{equation}
        \label{eq:bound-cond}
        \begin{aligned}
            \lim_{x_1 \to 0^+, x_1 \to +\infty} \w_k(x) p(x) \partial_k\log q(x) & = 0 \\
            \lim_{x_k \to a_k, x_k \to b_k} \w_k(x) p(x) \partial_k\log q(x) & = 0, \quad \forall k \geq 2,
        \end{aligned}
    \end{equation}
    then, $\mathbb{E}_{x \sim p} \operatorname{div}\left[ (\w \odot \nabla \log q) (x)\right]$ can be rewritten as
    \begin{align*}
        - \mathbb{E}_{x \sim p} \langle (\w^{1/2} \odot \nabla  \log q) (x), (\w^{1/2} \odot \nabla  \log p)(x) \rangle,
    \end{align*}
    where $\nabla$ denotes the gradient taken with respect to $x$.
\end{lemma}

\begin{proof}
    Denote $\mathcal{X}_k$ as its $k$-th marginalized dimension, $L(\cdot), U(\cdot): 2^\mathbb{R} \to \mathbb{R}$ as the mapping that extracts lower and upper bounds of a set, and $x_{-k}$ as the vector $x$ after removing its $k$-th entry. We make the following expansion:
    \allowdisplaybreaks
    \begin{align*}
        &\mathbb{E}_{x \sim p} \operatorname{div} [ (\w \odot \log q) (x) ]\\
        & = \int_{\mathcal{X}} p(x) \sum_{k = 1}^d \partial_k (\w_k(x) \partial_k \log q (x)) \d x \\
        & = \sum_{k = 1}^d \int_{\mathcal{X}_{-k}} \int_{\mathcal{X}_k} p(x) \partial_k (\w_k(x) \partial_k \log q (x)) \d x_k \d x_{-k} \\
        & = \sum_{k = 1}^d \int_{\mathcal{X}_{-k}} \left[ p(x) \w_k(x) \partial_k \log q(x) \Big|_{L(\mathcal{X}_k)}^{U(\mathcal{X}_k)} \right. \\
        & \quad \left. - \int_{\mathcal{X}_k} \partial_k p(x) \w_k(x) \partial_k \log q (x) \d x_k \right] \d x_{-k} \\
        & = - \sum_{k = 1}^d \int_{\mathcal{X}_{-k}} \int_{\mathcal{X}_k} \partial_k p(x) \w_k(x) \partial_k \log q (x) \d x_k \d x_{-k} \\
        & = - \int_{\mathcal{X}} \sum_{k = 1}^d \partial_k p(x) \w_k(x) \partial_k \log q (x) \d x \\
        & = - \int_{\mathcal{X}} \sum_{k = 1}^d \partial_k \log p(x) \cdot p(x) \cdot \w_k(x) \partial_k \log q (x) \d x \\
        & = - \mathbb{E}_{x \sim p} \left[ \sum_{k = 1}^d \partial_k \log p(x) \w_k(x) \partial_k \log q (x) \right] \\
        & = - \mathbb{E}_{x \sim p} \langle (\w^{1/2} \odot \nabla \log q) (x), (\w^{1/2} \odot \nabla \log p)(x) \rangle.
    \end{align*}
    Where the first equality is by definition of the divergence operator; the third equality is by the integral-by-parts trick; the fourth equality is by \eqref{eq:bound-cond}; the sixth equality is due to $\partial_k \log p(x) = \partial_k p(x) / p(x)$. The proof is completed.
\end{proof}

\begin{proof}[Proof of \Cref{thm:drift}]
    We only prove for $\mathbb{E}_0[\Delta(x)]$, and the proof for $\mathbb{E}_1[\Delta(x)]$ is similar.
    By definition of $\Delta(x)$, there is
    $$
    \mathbb{E}_0[\Delta(x)] = \mathbb{E}_0[\psi_0(x) - \psi_1(x)].
    $$
    By \eqref{eq:hyvarinen}, each term $\mathbb{E}_0[\psi_i(x)]$ is equal to:
    \begin{align*}
        \mathbb{E}_0 \left\| \sqrt{w(x)} \odot f_i (x) \right\|_2^2 + \underbrace{2 \cdot\mathbb{E}_0 \operatorname{div} \Big[ \w(x) \odot f_i (x) \Big]}_{(A)}
    \end{align*}
    Using \Cref{lem:bound-cond}, by taking $\nabla \log q \leftarrow f_i$ and $\nabla \log p \leftarrow f_0$,
    $$
    (A) = - 2 \cdot \mathbb{E}_0 \left\langle \underbrace{\sqrt{w(x)} \odot f_i(x)}_{(B)}, \underbrace{\sqrt{w(x)} \odot f_0(x)}_{(C)} \right\rangle.
    $$
    Therefore, there is
    \begin{align*}
        \mathbb{E}_0[\psi_i(x)]
        & = \mathbb{E}_0\left[ \| (B) \|^2_2 - 2 \langle (B), (C) \rangle + \| (C) \|_2^2 - \| (C) \|_2^2 \right]\\
        & = \mathbb{E}_0 \left[ \left\| \sqrt{w(x)} \odot \left[ f_i(x) - f_0(x) \right] \right\|_2^2\right] - \mathbb{E}_0\|(C)\|_2^2.
    \end{align*}
    Similarly, there is
    The first term vanishes when $i = 0$, and since for both $i = 0, 1$, $\mathbb{E}_0[\psi_i(x)]$ shares the second term and would thus cancel out. Therefore, there is
    $$
    \mathbb{E}_0[\Delta(x)] = \mathbb{E}_0 \left[ \left\| \sqrt{w(x)} \odot \left[ f_1(x) - f_0(x) \right] \right\|_2^2\right].
    $$
    Similarly to the proof above, one can also derive
    $$
    \mathbb{E}_1[\Delta(x)] = - \mathbb{E}_1 \left[ \left\| \sqrt{w(x)} \odot \left[ f_1(x) - f_0(x) \right] \right\|_2^2\right].
    $$
    By defining $D_F(f_0 \| f_1)$ and $D_F(f_1 \| f_0)$ as stated in \Cref{thm:drift}, we complete the proof.
\end{proof}

\subsection{Proof of \Cref{lem:param}}

\begin{proof}
    By \Cref{thm:drift}, we begin by expanding on $f_i(x)$:
    \begin{align*}
        f_i(x)
        & = \underbrace{\nabla_x \log \lambda_i(x)}_{(A)} - \underbrace{\nabla_x \int_{[t_n, t) \times \mathcal{B}(x)} \lambda_i(x') \d x'}_{(B)}.
    \end{align*}
    We analyze these two terms separately. For the first term:
    \begin{align*}
        (A)
        & = \nabla_x \lambda_i(x) / \lambda_i(x) = \frac{\alpha \nabla_x \int_{[0, t_n) \times \mathcal{B}(x)} \kappa(x, x') \d \mathbb{N}(x')}{\lambda_i(x)},
    \end{align*}
    where we used the fact that the base rate is a constant. 
    Since $\mathcal{B}(x)$ is a translation of a $\ell_p$-norm ball with radius $\delta$ by $x$, the right-hand side can be further expanded by using the Reynolds transport theorem \cite{marsden2003vector} as:
    $$
    (A) = \frac{\alpha \int_{[0, t_n) \times \mathcal{B}(x)} \nabla_x \kappa(x, x') + \nabla_{x'} \kappa(x, x') \d \mathbb{N}(x')}{\lambda_i(x)}.
    $$
    For term $(B)$, its first (temporal) entry is:
    \begin{align*}
        [(B)]_1
        & = \partial_t \int_{[t_n, t) \times \mathcal{B}(x)} \lambda_i(x') \d x'
        = \int_{\mathcal{B}(x)} \lambda_i(\tilde x') \d s' \\
        & = \mu_i \cdot |\mathcal{B}_\delta| + \underbrace{\alpha\int_{\mathcal{B}(x)} \int_{[0, t_n) \times \mathcal{B}(x)} \kappa(\tilde x', x'') \d \mathbb{N}(x'') \d s'}_{(C)},
    \end{align*}
    where we denote $\tilde x' \coloneq (t, s')$ for notation convenience.
    The other (spatial) entries are:
    \allowdisplaybreaks
    \begin{align*}
        [(B)]_{2:d}
        & = \nabla_s \int_{[t_n, t) \times \mathcal{B}(x)} \lambda_i(x') \d x' \\
        & = \int_{[t_n, t) \times\mathcal{B}(x)} \nabla_{s'} \lambda_i(x') \d s' \\
        & = \underbrace{\alpha\int_{[t_n, t) \times \mathcal{B}(x)} \int_{[0, t_n) \times \mathcal{B}(x)} \nabla_{s'} \kappa(x', x'') \d \mathbb{N}(x'') \d s'}_{(D)},
    \end{align*}
    where we also use the Reynolds transport theorem for the second equality. Since $(C)$ and $(D)$ are irrelevant to $i$ ($\kappa$ is shared), and therefore they cancel out each other in $f_1(x) - f_0(x)$. Denote the numerator of $(A)$ as
    $$
    \boldsymbol k(x) = \int_{[0, t_n) \times \mathcal{B}(x)} \nabla_x \kappa(x, x') + \nabla_{x'} \kappa(x, x') \d \mathbb{N}(x'),
    $$
    and by the fact that the pre- and post-change processes share the same form of the kernel $\kappa(\cdot, \cdot)$, there is:
    $$
    \frac{1}{\lambda_1(x)} - \frac{1}{\lambda_0(x)} = \frac{\mu_0 - \mu_1}{\lambda_1(x) \lambda_0(x)},
    $$
    therefore we conclude that
    there is
    $$
    f_1(x) - f_0(x) = (\mu_0 - \mu_1) \cdot \left[ \frac{\alpha \boldsymbol k(x)}{\lambda_1(x) \lambda_0(x)} +
    \begin{pmatrix}
        4\delta^2 \\
        \mathbf{0}_2
    \end{pmatrix}
    \right].
    $$
\end{proof}

\subsection{Proof of \Cref{cor:kernel}}

We provide proofs for each of the two settings.

\begin{proof}[Proof for HPP Case]
    Since the vector term $\boldsymbol k(x)$ defined in \Cref{lem:param} is zero under the HPP assumption, therefore,
    $$
    [f_1(x) - f_0(x)]_1 = (\mu_0 - \mu_1) \cdot |\mathcal{B}|,
    $$
    and all other entries of this vector are zero.
    By \Cref{thm:drift},
    \allowdisplaybreaks
    \begin{align*}
        D_F(f_i \| f_{1 - i})
        & = 
        \mathbb{E}_i \left\|
        \sqrt{\w(x)} \odot \left(
        f_0 (x) - f_1 (x)
        \right)
        \right\|_2^2 \\
        & = \mathbb{E}_i[t - t_n] \cdot \left( [f_0(x) - f_1(x)]_1 \right)^2 \\
        & = \frac{1}{\mu_i \cdot |\mathcal{B}|} \cdot (\mu_0 - \mu_1)^2 \cdot |\mathcal{B}|^2 \\
        & = \frac{(\mu_0 - \mu_1)^2 \cdot |\mathcal{B}|}{\mu_i},
    \end{align*}
    where the third equality is by the interarrival time of HPP.
    Therefore, we conclude that
    $$
    \mathbb{E}_i \left[ \Delta(x) \right] 
    = \text{sgn}\{i = 1\} \cdot \frac{(\mu_0 - \mu_1)^2 |\mathcal{B}|}{ \mu_i}.
    $$
\end{proof}

\begin{proof}[Proof for Kernel Case]

We first observe that $\boldsymbol{k}(x)$ vanishes under the exponential kernel. In fact, this property holds more generally for all \emph{translation-invariant} kernels. Specifically, recall that a kernel $\kappa : \mathcal X \times \mathcal X \to \mathbb R$ is translation invariant if
\begin{equation}
    \label{eq:kernel-inv}
    \kappa(x+a, x'+a) = \kappa(x,x') \quad \forall x,x',a \in \mathcal X.
\end{equation}
Equivalently, there exists a function $\varphi : \mathcal X \to \mathbb R$ such that
\[
\kappa(x,x') = \varphi(x - x').
\]
Under this representation, applying the chain rule yields
\[
\nabla_x \kappa(x,x') = \nabla \varphi(x - x'), 
\quad
\nabla_{x'} \kappa(x,x') = -\nabla \varphi(x - x').
\]
Therefore,
\[
\nabla_x \kappa(x,x') + \nabla_{x'} \kappa(x,x') \equiv 0,
\]
which proves the claim.
The condition~\eqref{eq:kernel-inv} is satisfied by a wide class of translation-invariant kernels, including the exponential kernel, Gaussian (radial) kernels, and the spatial triggering kernels used in ETAS models.

Back to the main proof, once again, we have
$$
[f_1(x) - f_0(x)]_1 = (\mu_0 - \mu_1) \cdot |\mathcal{B}|.
$$
Using the well-known result for the expected inter-event arrival time of a Hawkes process with an exponential excitation kernel,
$$
\mathbb{E}_i[t - t_n] = \frac{1 - \alpha}{\mu_i \cdot |\mathcal{B}| },
$$
we proceed analogously to the proof of the first statement and obtain
$$
\mathbb{E}_i \left[ \Delta(x) \right] 
= \text{sgn}\{i = 1\} \cdot \frac{(\mu_0 - \mu_1)^2 |\mathcal{B}|}{ \mu_i} \cdot \left( 1 - \alpha\right) .
$$
\end{proof}

\subsection{Proof for \Cref{thm:performance}}
\label{app:performance}


Since our setting implies non-i.i.d. data, we adapt the proof argument of Theorem 4 of \cite{lai1998information} for deriving our EDD bound. Following \cite{lai1998information}, we make assumptions on the information rate of the statistic:
\begin{ass}
    \label{ass:lln}
    Under the probability measure induced by $\lambda_1$,
    \begin{equation}
        \label{eq:info-rate}
        \lim_{t \to +\infty} \esssup_{\mathcal{H}_\tau} \frac{1}{N(t)} \int_{[\tau, t) \times \Omega} \Delta(x) \d \mathbb{N}(x) \overset{p}{\to} \mathbb{E}_1[\Delta(x)],
    \end{equation}
    \begin{equation}
        \label{eq:lln}
        \text{and} \quad
        \frac{N(t)}{t} \overset{p}{\to} \bar\lambda_1\frac{|\Omega|}{|\mathcal S|}.
    \end{equation}
\end{ass}
\Cref{ass:lln} states that asymptotically in the post-change regime, each post-change event roughly contributes the same amount of evidence on average \eqref{eq:info-rate}, and the number of these events roughly grows linear in time with a stable average \eqref{eq:lln}.
Both can be viewed as a type of law-of-large-numbers-type concentration requirement for the counting process, and hold trivially for HPPs due to the central limit theorem.
There have also been works \cite{hawkes1974cluster} showing similar properties for Hawkes processes with certain assumptions on its self-excitation kernel are met, such as stationarity and light-tailedness. 

\begin{proof}[Proof for EDD]
We fix an arbitrary change point $\tau \in [0, T)$ and work under the post-change probability measure induced by $\lambda_1$. For notational simplicity, denote
$$
I \coloneqq \mathbb{E}_1[\Delta(x)] = D_F(f_0\|f_1) > 0,
\quad
r \coloneqq \bar\lambda_1\frac{|\Omega|}{|\mathcal S|}.
$$

\paragraph*{Step 1: Constructing a proxy process}
We define a (non-maximized) cumulative evidence process
\[
\tilde W_{0:t} \coloneqq \int_{[0,t)\times\Omega} \Delta(x)\, d\mathbb N(x),
\]
and its threshold crossing time
\[
\tilde\nu \coloneqq \inf\bigl\{t\in[0,T): \tilde W_{0:t}\ge \gamma\bigr\}.
\]
By definition of $W_t$ as a supremum over time windows and spatial regions, we have $W_t\ge \tilde W_{0:t}$ for all $t$, hence
\[
\nu = \inf\{t: W_t\ge\gamma\} \le \inf\{t: \tilde W_{0:t}\ge\gamma\} = \tilde\nu,
\]
and therefore it suffices to bound $\mathbb E_1[\tilde\nu]$, since $\mathbb E_1[\nu]\le \mathbb E_1[\tilde\nu]$.

\paragraph*{Step 2: One-block underflow probability.}
Fix $\varepsilon\in(0,I)$ and define the block length
\[
t' \coloneqq \frac{\gamma}{(I-\varepsilon)\,r}.
\]
Note that $\{\tilde\nu>t'\}\subseteq\{\tilde W_{0:t'}<\gamma\}$, hence
\begin{align*}
    \mathbb P_1(\tilde\nu>t')
    &\le \mathbb P_1 \left(\tilde W_{0:t'}<\gamma\right) \\
    & = \mathbb P_1 \left(\frac{\tilde W_{0:t'}}{t'} < \frac{\gamma}{t'}\right)
    = \mathbb P_1 \left(\frac{\tilde W_{0:t'}}{t'} < (I-\varepsilon) r\right).
\end{align*}
By Assumption~\ref{ass:lln}, we have the convergence in probability
\begin{align*}
    & \frac{\tilde W_{0:t}}{t} = \frac{1}{t}\int_{[0,t)\times\Omega} \Delta(x)\, d\mathbb N(x) \\
    & = \Bigl(\frac{1}{N(t)}\int_{[0,t)\times\Omega}\Delta(x) d\mathbb N(x)\Bigr)
    \cdot
    \Bigl(\frac{N(t)}{t}\Bigr)
    \overset{p}{\longrightarrow}
    Ir,
\end{align*}
where the product convergence follows from Slutsky's theorem.
Consequently,
\[
\mathbb P_1(\tilde\nu>t') = o(1)\qquad\text{as }\gamma\to\infty.
\]

Moreover, by the $\esssup_{\mathcal H_\tau}$ form in~\eqref{eq:info-rate}--\eqref{eq:lln}, the same argument yields the \emph{uniform} conditional version:
there exists $\gamma_0(\varepsilon)$ such that for all $\gamma\ge \gamma_0(\varepsilon)$,
\[
c \coloneqq \sup_{u\ge 0}\ \esssup_{\mathcal H_u}\mathbb P_1\left(
\tilde W_{u: u + t'} < \gamma
\ \middle|\ \mathcal H_u
\right)
<1.
\]
There is moreover $c=c(\gamma,\varepsilon)\to 0$ as $\gamma\to\infty$.

\paragraph*{Step 3: Geometric tail bound across blocks.}
For integers $j\ge 1$, using the tower property and the definition of $c$,
\begin{align*}
\mathbb P_1(\tilde\nu>jt')
&=\mathbb E_1\left[\mathbf 1\{\tilde\nu>(j-1)t'\}\,
\mathbb P_1\left(\tilde\nu>jt'\mid \mathcal H_{(j-1)t'}\right)\right] \\
&\le \mathbb E_1\Bigl[\mathbf 1\{\tilde\nu>(j-1)t'\} \times \\
& \quad \left. \mathbb P_1\left(
\tilde W_{(j-1)t':jt'}
< \gamma
\ \middle|\ \mathcal H_{(j-1)t'}
\right)\right] \\
&\le c\,\mathbb P_1(\tilde\nu>(j-1)t') \le  \ldots \le c^j,
\end{align*}
where the last inequality follows from recursion.

\paragraph*{Step 4: Tail integral identity and expectation bound.}
Using the tail integral identity and partitioning time into blocks,
\begin{align*}
\mathbb E_1[\tilde\nu]
&=\int_0^\infty \mathbb P_1(\tilde\nu>t)\,dt
=\sum_{j=0}^\infty \int_{jt'}^{(j+1)t'} \mathbb P_1(\tilde\nu>t)\,dt \\
&\le \sum_{j=0}^\infty t'\,\mathbb P_1(\tilde\nu>jt')
\le t'\sum_{j=0}^\infty c^j
= \frac{t'}{1-c}.
\end{align*}
Since $c\to 0$ as $\gamma\to\infty$, we have $(1-c)^{-1}=1+o(1)$ and thus
\[
\mathbb E_1[\tilde\nu]
\le
\frac{\gamma}{(I-\varepsilon)\,r}\,(1+o(1)).
\]
Letting $\varepsilon\downarrow 0$ yields
\[
\mathbb E_1[\tilde\nu]
\le
\frac{\gamma}{Ir}\,(1+o(1))
=
\frac{\gamma\,|\mathcal S|}{\bar\lambda_1|\Omega|\,D_F(f_0\|f_1)}\,(1+o(1)).
\]
Finally, since $\nu\le \tilde\nu$, we conclude
\[
\mathbb E_1[\nu]
\le
\frac{\gamma\,|\mathcal S|}{\bar\lambda_1|\Omega|\,D_F(f_0\|f_1)}\,(1+o(1)),
\]
and notice that the EDD must always be upper bounded by the time horizon $T$,
which establishes the desired EDD bound.
\end{proof}

To prove the Jaccard index statement, we require the following lemma, where we refer to the localized process formed by $\mathcal{B}_\delta(x)$ as a $\mathcal{B}_\delta(x)$-process, and refer to the next event in this process as the next $\mathcal{B}_\delta(x)$-event.

\begin{lemma}
    \label{lem:B-cond-pdf}
    Given intensity $\lambda(x)$ defined on $\mathcal{B}_\delta(x)$, the conditional PDF of the next $\mathcal{B}_\delta(x)$-event is defined as:
    \begin{equation}
        \label{eq:cond-b-pdf}
        p(x; \delta) =
        \lambda(x) \exp\left(
        - \int_{[t_n, t(x)) \times \mathcal{B}_\delta(x)}
        \lambda(x')\d x'
        \right).
    \end{equation}
\end{lemma}

\begin{proof}[Proof for \Cref{lem:B-cond-pdf}]
    For $x$ to be the next $\mathcal{B}_\delta(x)$-event, three conditions must be satisfied:
    \begin{enumerate}
        \item An event occurs at $x$.
        \item There are no event occurrences within $[t_n, t(x)) \times \mathcal{B}_\delta(x)$.
        \item $s(x) \in \mathcal{B}_\delta(x)$
    \end{enumerate}
    Denote these three events by $A$, $B$, $C$ accordingly, then there is
    \begin{align*}
        & p(x ; \delta) = \lim_{\d x \to 0} \frac{\mathbb{P}(A \cap B \cap C)}{\d x}
        = \\
        & \lim_{\d x \to 0} \frac{\mathbb{P}(A | B )}{\d x} \cdot \mathbb{P}(B) \cdot \mathbbm{1}\{s(x) \in \mathcal{B}_\delta(x)\}, 
    \end{align*}
    where the second equality is by definition of $C$. Note that since $B \subseteq \mathcal{H}_{t(x)}$, the first term is:
    $$
    \lim_{\d x \to 0} \frac{\mathbb{P}(A | B )}{\d x}
    = \lim_{\d x \to 0} \frac{\mathbb{E}[\mathbb{N}( \d x) | \mathcal{H}_{t(x)}]}{\d x} =  \lambda(x).
    $$
    To solve for the second term $\mathbb{P}(B)$, we define the survival function $S(t)$ as
    $$
    S(t) \coloneqq \mathbb{P}(\text{No event in } [0, t) \times \mathcal{B}_\delta(x)),
    $$
    \begin{align*}
        & \lim_{\d t \to 0} \frac{\mathbb{P}(\text{No event in } [t, t + \d t) \times \mathcal{B}_\delta(x))}{\d t} \\
        & \quad \quad =  1 - \Lambda(t) \lim_{\d t \to 0} \d t + \lim_{\d t \to 0} o(\d t)),
    \end{align*}
    where $\Lambda(t)$ is the cumulative hazard defined as
    $
    \Lambda(t) = \int_{\mathcal{B}_\delta(x)} \lambda(x) \d s.
    $
    Then, by definition, there is
    \begin{align*}
        & \lim_{\d t \to 0} S(t + \d t) \\
        & = S(t) \cdot \lim_{\d t \to 0} \frac{\mathbb{P}(\text{No event in } [t, t + \d t) \times \mathcal{B}_\delta(x))}{\d t} \\
        & = S(t) - \lim_{\d t \to 0} \Lambda(t) \cdot S(t) \cdot \d t + \lim_{\d t \to 0} o(\d t).
    \end{align*}
    Therefore, there is the following ordinary differential equation (ODE):
    $$
    \frac{\d}{\d t}S(t) = - \Lambda(t) \cdot S(t), \quad S(0) = 1.
    $$
    Solving for it, and we yield:
    $$
    S(t) = \exp \left( -\int_{0}^t \Lambda(t) \d t\right) = \exp \left( -\int_{[0, t) \times \mathcal{B}_\delta(x)} \lambda(t) \d t\right).
    $$
    Therefore, 
    $$
    \mathbb{P}(B) = S(t(x)) - S(t_n) = \exp \left( -\int_{[t_n, t(x)) \times \mathcal{B}_\delta(x)} \lambda(t) \d t\right),
    $$
    where $t_n$ denotes the time of the last $\mathcal{B}$-event observed prior to $x$.
    Combining the above derivation, we conclude that
    $$
    p(x; \delta) =
    \lambda(x) \exp\left(
    - \int_{[t_n, t(x)) \times \mathcal{B}_\delta(x)}
    \lambda(x')\d x'
    \right).
    $$
    where we used the fact that $s(x) \in \mathcal{B}_\delta(x)$ is always true.
    This finishes the derivation.
\end{proof}

\begin{proof}[Proof of Jaccard Index]
    Define the set dilation and erosion operators as
    \begin{equation}
        \label{eq:dilation}
        \Omega \oplus \mathcal{B} \coloneqq \cup_{s \in \Omega}\{ s' : s' \in \mathcal{B}(s) \}
    \end{equation}
    \begin{equation}
        \label{eq:erosion}
        \Omega \ominus \mathcal{B} \coloneqq \cup_{s \in \mathcal{R}} \{ s  :  \mathcal{B}(s) \subseteq \mathcal{R}\}
    \end{equation}
    and define the following set of sets
    $$
    \mathcal{I} = \left\{ \Omega' : (\Omega \ominus \mathcal{B}) \subseteq \Omega' \subseteq (\Omega \oplus \mathcal{B}) \right\}. 
    $$
    Define the \emph{definitely post-change} and \emph{definitely pre-change} regions as
    \[
    \Omega_{\mathrm{in}} \coloneqq \Omega \ominus B_\delta
    = \{s \in S : B_\delta(s) \subseteq \Omega\},\]
    \[
    \Omega_{\mathrm{out}} \coloneqq (\Omega \oplus B_\delta)^c
    = \{s \in S : B_\delta(s) \subseteq \Omega^c\}.
    \]
    By construction, events whose spatial coordinates lie in $\Omega_{\mathrm{in}}$ ($\Omega_{\mathrm{out}}$) are generated entirely under the post-change (pre-change) regime within their local neighborhoods. Under the properness of the score statistic (\Cref{thm:drift}),
    \[
    \mathbb E_\tau[\Delta(x) \mid s(x) \in \Omega_{\mathrm{in}}] > 0,
    \quad
    \mathbb E_\tau[\Delta(x) \mid s(x) \in \Omega_{\mathrm{out}}] < 0.
    \]
    By ergodicity of the post-change point process and the law of large numbers (\Cref{ass:lln}), for any measurable set $A \subseteq S$, as $t\to\infty$,
    \begin{equation}
        \label{eq:lln-jaccard}
        \frac{1}{t}\int_{[0,t)\times A} \Delta(x) \d \mathbb N(x)
        \xrightarrow{\mathrm{a.s.}}
        \bar \lambda_1 |A|  \mathbb E_\tau[\Delta(x)| s(x)\in A],
    \end{equation}
    
    Suppose, for contradiction, that the maximizer $\hat\Omega$ includes a subset
    $A \subseteq \Omega_{\mathrm{out}}$ with $|A|>0$. Then, according to \eqref{eq:lln-jaccard}, the contribution from
    $[0,t)\times A$ is asymptotically negative, and removing $A$ strictly increases
    the accumulated statistic for all sufficiently large $t$, contradicting the
    optimality of $\hat\Omega$. Similarly, if $\hat\Omega$ excludes a subset
    $A \subseteq \Omega_{\mathrm{in}}$ with $|A|>0$, then adding $A$ contributes
    positive drift and strictly increases the statistic, again contradicting
    optimality. Therefore, with probability tending to one as $\gamma\to\infty$,
    \[
    \Omega \ominus B_\delta \subseteq \hat\Omega \subseteq \Omega \oplus B_\delta,
    \]
    which is equivalent to $\hat \Omega \in \mathcal I$.
    Therefore, we conclude that
    \begin{align*}
        & \lim_{\gamma \to +\infty}J(\hat \Omega, \Omega)
        \geq
        \lim_{\gamma \to +\infty} \inf_{\Omega' \subseteq \mathcal I}J(\Omega', \Omega) \\
        & \geq \min \left\{ \frac{|\Omega|}{|\Omega \oplus \mathcal{B}_\delta|}, \right. \left. \frac{|\Omega \ominus \mathcal{B}_\delta|}{|\Omega|} \right\}.
    \end{align*}
\end{proof}

\subsection{Additional Experiment Details}
\label{app:exp}

\begin{algorithm}[t]
\caption{Online Algorithm}
\label{alg:online}
\begin{algorithmic}[1]
    \REQUIRE Reference dataset  $\mathcal{D}_{0}$;
    Test dataset $\mathcal{D}$;
    Current time $t$;
    Optimization epochs $K$;
    Threshold $\gamma$;
    Learning rate $\eta$.
    \STATE  Train $f_{0}$.
    \STATE $f_1 \leftarrow f_0$ Initialize post-change model . 
    \STATE $\hat \nu^{(0)} \leftarrow 0$.
    \FOR{$x_j \in \mathcal{D}$}
        \STATE Compute $\psi_0(x_j)$ and $\psi_1(x_j)$ as defined in \eqref{eq:hyvarinen}.
        \STATE $\Delta(x_j) \leftarrow \psi_0(x_j) - \psi_1(x_j)$.
        \STATE $\hat \nu^{(0)}_j \leftarrow 0$.
        \FOR{$k \in \{ 1, \ldots, K \}$}
            \STATE $\hat \Omega^{(k)}_j \leftarrow$ Compute change region as in \eqref{eq:est-R}.
            \STATE $\hat \tau^{(k)}_j \leftarrow$ Compute change-point as in \eqref{eq:est-nu}.
        \ENDFOR
        \STATE Update $f_1$ via gradient descend as in \eqref{eq:grad-dec}.
    \ENDFOR
    \STATE Set $\hat \tau \leftarrow \hat \tau^{(k)}$ and $\hat \Omega \leftarrow \hat \Omega^{(k)}$.
    \STATE Compute spatio-temporal CUSUM statistic:
    $$
    W_t \leftarrow \int_{[\hat \tau, t) \times \hat \Omega} \Delta(x) \d \mathbb{N}(x).
    $$
    \IF{$W_t \geq \gamma$}
        \STATE $\nu \leftarrow t$ Set stopping time.
    \ENDIF 
    \RETURN Estimation $(\hat \tau, \hat \Omega)$; Stopping time $\nu$.
\end{algorithmic}
\end{algorithm}

\begin{figure}[t]
    \centering
    \includegraphics[width=0.8\linewidth]{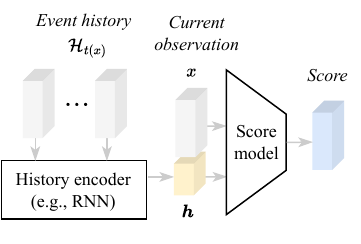}
    \caption{Architecture of the score model with conditional history information. A history encoder maps the given event history into a latent vector $\boldsymbol h$, which is concatenated with the current observed event $x$, and fed into a feed-forward neural network to produce the final score estimate.
    }
    \label{fig:architecture}
\end{figure}


\subsubsection{Threshold Simulation}
\label{app:sim}

The threshold \(\gamma\) is typically chosen to balance the trade-off between the probabilities of false alarms and successful detection, as determined by a target ARL value. To reduce computational effort in determining thresholds for large target ARL values, we employ an efficient approximation algorithm. This algorithm leverages the fact that the stopping time \(\nu\) under the pre-change regime approximately follows an exponential distribution when the ARL is large. Such approximation methods are widely utilized in online change detection.

The high-level idea of the procedure is that, instead of simulating the mean of the distribution of $T:= \inf\{t: W_t \geq \tau\}$ directly, we obtain an estimate of the mean from an estimate of the cumulative distribution function of $T$ based on $N_1$ iterations.
Specifically, in each iteration, we simulate the pre-change trajectory with $N_2$ time steps, and compute the maximum of the detection statistics at these $N_2$ time steps. These maximum values under $N_1$ iterations are then denoted as $W_{1,\text{max}}, W_{2,\text{max}}, \ldots, W_{N_1,\text{max}}$. For the desired ARL values $\gamma$, we approximate the stopping time $T$ as an exponential distribution with mean $\gamma$. Thus we have $P(W_{\text{max}} < \gamma ) = P(\nu > N_2)  \approx e^{-N_2/\gamma}$. Thus the corresponding threshold $\tau$ can be approximated as the $e^{-N_2/\gamma}$ quantile of the set \(\{W_{1,\text{max}}, W_{2,\text{max}}, \ldots, W_{N_1,\text{max}}\}\). 

Note that to adapt to the point process setting, we make an additional modification here by replacing the number of previous timesteps with the number of observed events when determining the significance level of the quantile, as suggested in \cite{wang2023sequential}.

\paragraph{Score model architecture}

\Cref{fig:architecture} illustrates the architecture used to parameterize the localized score model $f_i(\cdot;\delta)$.
Each observed event $x = (t, s)$ is first represented by its local spatio-temporal context, consisting of the inter-arrival time $t - t_n$, the spatial coordinates $s$, and the set of historical events falling within the spatial neighborhood $B_\delta(x)$.
These locally preceding events are sequentially encoded by a recurrent neural network (RNN), which summarizes the truncated event history into a fixed-dimensional latent representation capturing local temporal dependence and self-excitation effects.
The resulting hidden state is then concatenated with features of the current event $x$ and passed through a feed-forward neural network that outputs an estimate of the score function $\nabla_x \log p_i(x;\delta)$.
This design follows \cite{dong2025conditional}.
In our experiments, we parameterized it by: 
($i$) a one-layer LSTM network with $32$ hidden units, which encodes event history, followed by
($ii$) a one-layer feed-forward network with $512$ hidden units, which takes as input the concatenation of the history embedding and the current observed event, outputting the score.

\end{document}